\documentclass[11pt]{article}
\usepackage{booktabs} 
\usepackage[ruled]{algorithm2e}

\usepackage{graphicx}
\usepackage{url}            
\usepackage{booktabs}       
\usepackage{amsfonts}       
\usepackage{nicefrac}       
\usepackage{amsmath}
\usepackage{amsthm}
\usepackage{natbib}
\usepackage{color}
\usepackage{bbm}
\usepackage{bm}
\usepackage{thm-restate}
\usepackage{hyperref}
\usepackage{amssymb}
\usepackage{subcaption}

\DeclareCaptionLabelFormat{andtable}{#1~#2  \&  \tablename~\thetable}
\usepackage[capitalize,noabbrev]{cleveref}

\definecolor{DarkGreen}{rgb}{0.1,0.5,0.1}
\definecolor{DarkRed}{rgb}{0.5,0.1,0.1}
\definecolor{DarkBlue}{rgb}{0.1,0.1,0.5}
\definecolor{Gray}{rgb}{0.2,0.2,0.2}

\definecolor{negvals}{rgb}{0.1,0.5,0.1}
\definecolor{posvals}{rgb}{0.5,0.1,0.1}
\definecolor{neutralvals}{RGB}{255, 140, 0}
\definecolor{neutralvals}{rgb}{0.6, 0.51, 0.48}
\definecolor{posvals}{rgb}{0.8, 0.0, 0.0}

\usepackage{booktabs}

\theoremstyle{plain}
\newtheorem{theorem}{Theorem}[section]
\newtheorem{lemma}[theorem]{Lemma}
\newtheorem{claim}[theorem]{Claim}

\newtheorem{corollary}[theorem]{Corollary}

\theoremstyle{definition}
\newtheorem{definition}[theorem]{Definition}

\theoremstyle{remark}

\usepackage{algpseudocode}

\DeclareMathOperator*{\E}{\mathbb{E}}

\newcommand{\Ind}{\mathbb I}

\newcommand{\defeq}{\stackrel{\small \mathrm{def}}{=}}

\DeclareMathOperator*{\argmin}{argmin}
\DeclareMathOperator*{\argmax}{argmax}

\newcommand{\cA}{{\mathcal A}}
\newcommand{\cB}{{\mathcal B}}
\newcommand{\cC}{{\mathcal C}}
\newcommand{\cD}{\mathcal{D}}
\newcommand{\cE}{{\mathcal E}}
\newcommand{\cF}{{\mathcal F}}
\newcommand{\cG}{{\mathcal G}}
\newcommand{\cH}{{\mathcal H}}

\newcommand{\cP}{{\mathcal P}}

\newcommand{\deltaIR}{\eta}

\let\emptyset\varnothing

\renewcommand \vec [1]{\bm{#1}}
\newcommand{\rev}{\textnormal{rev}}

\newcommand{\LB}{\textnormal{LB}}

\DeclareFontFamily{U}{mathx}{}
\DeclareFontShape{U}{mathx}{m}{n}{<-> mathx10}{}
\DeclareSymbolFont{mathx}{U}{mathx}{m}{n}
\DeclareMathAccent{\widehat}{0}{mathx}{"70}
\DeclareMathAccent{\widecheck}{0}{mathx}{"71}

\usepackage{defns}
\usepackage{soul}
\usepackage{fullpage}

\title{Leveraging Reviews:\\Learning to Price with Buyer and Seller Uncertainty}

\author{
	Wenshuo Guo\thanks{Work done while the author was a PhD student at UC Berkeley.} \\ UC Berkeley \\ \texttt{wsguo@berkeley.edu} \and 
	Nika Haghtalab \\ UC Berkeley \\ \texttt{nika@berkeley.edu}
    \and
	Kirthevasan Kandasamy \\ UW Madison \\ \texttt{kandasamy@cs.wisc.edu }
  \and
	Ellen Vitercik \\ Stanford University \\ \texttt{vitercik@stanford.edu}}

\begin{document}

\maketitle

\begin{abstract}
In
online marketplaces, customers have access to hundreds of reviews for a single product.  Buyers
often use reviews from other customers that share their type---such as height for clothing, skin
type for skincare products, and location for outdoor furniture---to estimate their values, which
they may not know \emph{a priori}. Customers with few relevant reviews may hesitate to make a
purchase except at a low price, so for the seller, there is a tension between setting high prices
and ensuring that there are enough reviews so that buyers can confidently estimate their values.
Simultaneously, sellers may use reviews to gauge the demand for items they wish to sell.

In this work, we study this pricing problem in an online learning setting where the seller interacts
with a set of buyers of finitely many types, one by one, over a series of $T$ rounds. At each round,
the seller first sets a price. Then a buyer arrives and examines the reviews of the previous buyers
with the same type, which reveal those buyers' \emph{ex-post} values. Based on the reviews, the
buyer decides to purchase if they have good reason to believe that their \emph{ex-ante} utility is
positive. Crucially, the seller does not know the buyer's type when setting the price, nor even the
distribution over types. We provide a no-regret algorithm that the seller can use to obtain high
revenue. When there are $d$ types, after $T$ rounds, our algorithm achieves a problem-independent
$\tilde O(T^{2/3}d^{1/3})$ regret bound. However, when the smallest probability $q_{\min}$ that any
given type appears is large, specifically when $q_{\min} \in \Omega(d^{-2/3}T^{-1/3})$, then the
same algorithm achieves a $\tilde O(T^{1/2}q_{\min}^{-1/2})$ regret bound. Our algorithm starts by
setting lower prices initially so as to \emph{(i)} boost the number of reviews and increase the
accuracy of future buyers' value estimates while also \emph{(ii)} allowing the seller to identify
which customers need to be targeted to maximize revenue. This mimics real-world pricing dynamics. We
complement these upper bounds with matching lower bounds in both regimes, showing that our algorithm
is minimax optimal up to lower-order terms.

\end{abstract}

\section{Introduction}

The rapid growth of e-commerce, now accounting for 22\% of global retail sales\footnote{\href{https://www.trade.gov/ecommerce-sales-size-forecast}{https://www.trade.gov/ecommerce-sales-size-forecast}}, has
allowed customers to make far more informed purchase decisions than ever before. Potential buyers
can gain insights from thousands of reviews before deciding whether to purchase an
item. Customers often use reviews by buyers who share their ``type''---such as body type for clothes
or skin type for skincare products---to develop high-fidelity estimates of how much they value
different items,
which are quantities they may be uncertain of before purchasing. 

When learning from reviews, a customer's purchase decision is no longer just a function of the item's price but also of how certain the customer is about her valuation, which in turn depends on the earlier sales and reviews of the items.
This leads to a tension between setting revenue-optimal prices while ensuring that buyers have enough reviews to confidently estimate their values.
This tension is perhaps most clear  for
customers of rare types (for example, particularly tall or short individuals shopping for clothing) who may find only a few reviews from similar customers and, due to this uncertainty, may only be willing to buy at relatively low prices.

We introduce a model that simultaneously captures the seller's pricing problem, the buyers' learning problem, and the modus through which the buyers learn: reviews. We study how a seller---who is uncertain about the buyers' type distribution---can learn to set high-revenue prices when the buyers themselves are uncertain about their own values and are learning from reviews.
Thus, there is information uncertainty on both sides of the market: the seller has uncertainty about which buyer will arrive and the buyers’ type distribution, but the buyer, who knows their type, suffers from the uncertainty about their \emph{ex-ante} value. Both sides of the market are operating with significantly less information than has historically been assumed in mechanism design.
We study this pricing problem with an online sequential learning model where the seller attempts to sell identical copies of an item to a series of distinct buyers over $T$ timesteps.
Each buyer has one of $d$ types drawn from a distribution $\cP$, and a buyer of type $i$ has an \emph{ex-ante} value of $\theta_i$ for the item.

At each timestep $t$, the seller sets a price $p_t$. Although the seller knows the \emph{ex-ante} values $\theta_1, \dots, \theta_d$  and thus has some limited information about the buyers (for example, from market research), he does not know the buyer's type on each round nor even the distribution $\cP$.
A buyer on any round could be of
\emph{(i)} a high-value type, but who is uncertain of their value since their type has few
reviews, and thus may be hesitant to make a purchase except at a low price,
\emph{(ii)} of a high-value type, and who is more certain of their value since their type has many
reviews, and thus is willing to purchase at a high price,
or
\emph{(iii)} of a low-value type whom the seller should not target even if they were absolutely
certain of their value since it leads to small per-purchase revenue.

If a buyer of type $i$ purchases the item, they will leave a review communicating their  \emph{ex-post} value for the item, which is a random variable with mean $\theta_i$. To decide whether to purchase, a new buyer evaluates reviews left by buyers of type $i$ who bought the item in the past. Specifically, the buyer at round $t \in [T]$ uses the past reviews to select a threshold $\tau_t$ and chooses to buy as long as $p_t \leq \tau_t$. If the buyer's threshold $\tau_t$ is too pessimistic---for example, it always equals zero no matter the reviews---then optimizing revenue would be hopeless. In our model, we bound the level of pessimism that the buyer can display: we assume that $\tau_t$ is at least a lower confidence bound we denote $\LB_t$ that equals the average of the reviews left by buyers with the same type, minus an uncertainty term that depends on the number of such reviews. Intuitively, the buyer can be confident that their \emph{ex-ante} value is at least $\LB_t$ with high probability, so they always buy if they have good reason to believe that their \emph{ex-ante} utility (value minus price) will be positive.

The \emph{ex-post} value is the actual experience of the buyer and is different from the \emph{ex-ante} value due to exogenous stochastic factors that cannot be known at the time of purchase (for example, manufacturing defects, color on the website not matching the actual color). Hence, the buyer decides based on their \emph{ex-ante} value when there is complete information. In our problem, the buyer does not even know their \emph{ex-ante} value and uses reviews from previous buyers (whose reviews are based on their actual experiences, i.e., \emph{ex-post} values) to update their estimate of the \emph{ex-ante} value (as the expected ex-post value is the \emph{ex-ante} value).

\subsection{Our contributions}
\label{sec:contributions}

We provide a no-regret learning algorithm for the seller that balances setting high-revenue prices with soliciting reviews from rare but high-value customers.

\paragraph{Key technical challenges.}

The seller does not know the current buyer's type on each
round \emph{a priori}, which means the prices are anonymous. Moreover, this means the seller
 does not know the number of reviews that the buyer will use to
construct their value estimate. If the buyer on round $t$ has a rare type, then the lower confidence
bound $\LB_t$ will be low, and thus the seller would have to set a low price to ensure a purchase and
a review. Suppose this rare type of buyer's \emph{ex-ante} value is high enough. In that case, it may be
worthwhile to initially set a low price to solicit enough reviews to ensure future
purchases at a higher price, thereby winning over these rare but high-value customers. The seller,
however, has to decide which buyers to win over without knowing the type of the buyer on each round,
nor even the distribution over types (and, thus, which types are common and which are rare). He
may, therefore, wastefully offer a low price to a high-value
buyer with a common type---meaning that $\LB_t$ is
near the buyer's \emph{ex-ante} value---who would be willing to buy at a higher price. If a rare buyer’s value is high enough, it may be worthwhile to set a low price to ensure future purchases at a higher price. However, if the buyer’s type is exceedingly rare, the seller will lose too much revenue by setting such a low price. The challenge is that the seller has to decide which buyers to target without knowing the distribution over types.

\paragraph{Algorithm overview.} With this intuition in mind, our algorithm maintains a set $S_t$ at each step $t$ consisting of buyer types with a sufficiently high value that are not exceedingly rare. It gradually refines this set over the $T$ rounds. Intuitively, $S_t$ is the set of buyers the algorithm targets. To refine $S_t$, the algorithm has two phases. In the first phase, the algorithm offers the item for free for a carefully chosen number of rounds, observing i.i.d. samples from the type distribution.
\begin{figure}[t]
\centering
    \begin{subfigure}[b]{0.5 \textwidth}
        \includegraphics[scale=1]{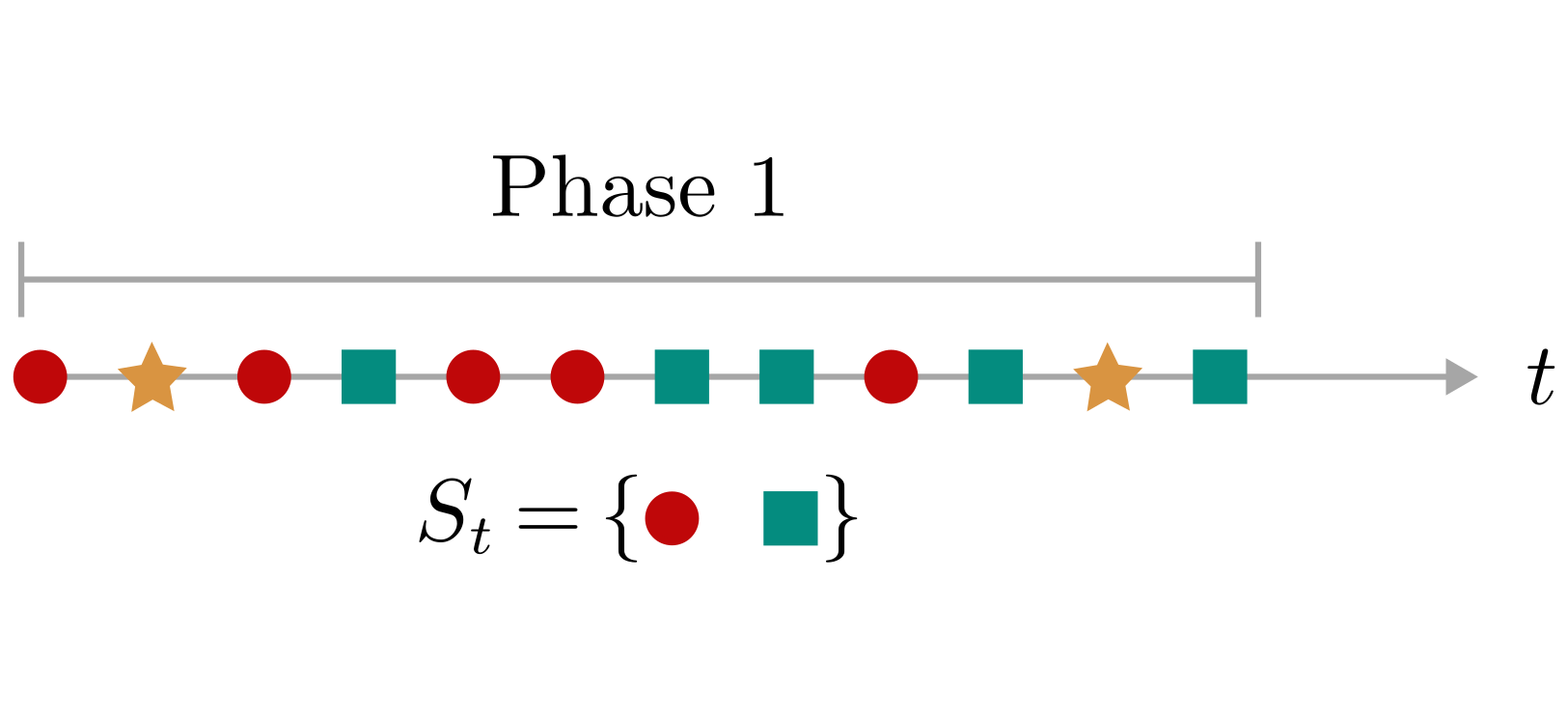}
        \centering
        \caption{Illustration of our algorithm's first phase, at the end of which only the red circle and the green square are in $S_t.$}
        \label{fig:phase1}
    \end{subfigure}
    \qquad
    \begin{subfigure}[b]{0.44\textwidth}
        \includegraphics[scale=1]{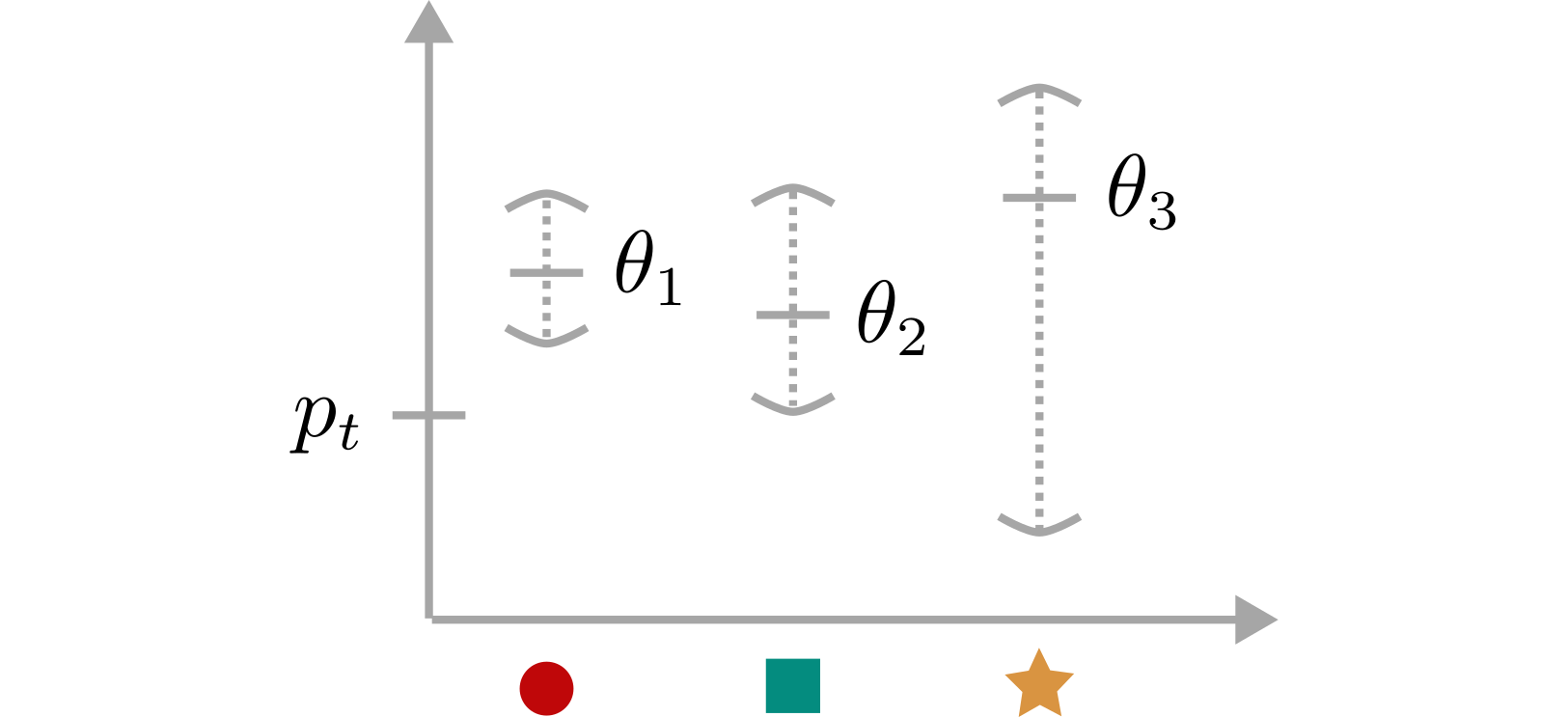}
        \centering
      \caption{During the second phase, the algorithm sets the price $p_t$ low enough to ensure types in $S_t$ will buy.}
        \label{fig:phase2}
    \end{subfigure}
    \caption{Illustration of our algorithm's first and second phases when there are three types: a red circle, green square, and orange star.}\label{fig:alg}
\end{figure}
The algorithm sets $S_t$ to be the set of types appearing in a sufficiently large fraction of rounds, as in Figure~\ref{fig:phase1}. In the second phase, the algorithm sets the price low enough to ensure that buyers in $S_t$ always buy the item, as in Figure~\ref{fig:phase2}. It successively eliminates types from $S_t$ that contribute too little revenue.

\paragraph{Regret upper bound and proof overview.} In this model, we define regret as the difference between (1) the algorithm's total expected revenue and (2) the expected revenue of the optimal fixed price if the buyers bought whenever their \emph{ex-ante} value was larger than the price, i.e., $\max p\Pr_{i \sim \cP}[\theta_i \geq p].$

We contend with several sources of regret. The first phase of the algorithm, where the item is sold for free, inevitably leads to regret, so it must be made as brief as possible.
The algorithm then completely disregards the buyer types that appeared too rarely during that phase. This results in a subset $Q \subseteq [d]$ of buyer types that appear sufficiently often. In the second phase, the algorithm only attempts to optimize revenue with respect to the buyers in $Q$ instead of the entire set $[d]$, which contributes to regret.
Finally, the buyers themselves do not know their \emph{ex-ante} values, whereas, under our regret benchmark, buyers buy whenever their \emph{ex-ante} value is larger than the price.

We obtain our final regret bound by analyzing these three sources of error.
Our bound depends on the smallest probability that any given type appears, which we denote as $q_{\min}.$ If $q_{\min}$ is not tiny---specifically, $q_{\min} > 2d^{-2/3}T^{-1/3}$---then we obtain a regret bound that scales with $\sqrt{T}$, as desired. In particular, our regret upper bound is $\tilde O(T^{1/2}q_{\min}^{-1/2} + T^{1/3}d^{2/3})$. Otherwise,
for arbitrary $\qmin$, our regret bound scales with $T^{2/3}$ as $\tilde O\left(T^{2/3}d^{1/3} + T^{1/3}d^{2/3}\right)$.

\paragraph{Regret lower bound and proof overview.}
Typical bandit lower bounds rely on hypothesis testing arguments to show that any algorithm would struggle to distinguish between similar problems but with different optimal outcomes. Such an analysis would not capture the main
difficulty in our setting: how fast customers can estimate their \emph{ex-ante} values from past reviews.
Instead, our proof leverages the buyers' uncertainty
to establish a  $\tilde\Omega\left(T^{2/3}d^{1/3}\right)$ worst-case lower bound and a $\tilde\Omega(T^{1/2}q_{\min}^{-1/2})$ lower bound
when $\qmin$ is large.
This establishes the optimality of our algorithm.

Our proof constructs a hard problem instance
where buyer types with low probability of appearance have comparable \emph{ex-ante} values to types with high probability of appearance.
On each round, an algorithm should decide whether it wishes to target low-probability customers
who may be less certain about their value due to fewer reviews and consequently have
small ${\rm LB}_t$.
Keeping prices low to do so leads to low revenue in the current round, but ignoring low-probability customers by choosing a high price risks losing potentially high per-purchase revenue in the future.
By carefully choosing the probability of appearance in our construction, we obtain a tight lower bound.

Our lower bound proof also provides insights that support the structure of our learning algorithm.
If the seller knew the type distribution, 
he could choose a threshold \emph{a priori}
and only target customer types with probability larger than that threshold.
Our proof illustrates that no policy could do essentially better than this thresholding approach: it does not help significantly to dynamically change which types the seller targets based on appearance probability. 
Our algorithm exhibits a similar behavior
even though the seller does not know the type distribution: it uses the first phase to discard low-probability types, focusing on the remaining types in the second phase.
\subsection{Related work}\label{sec:relateda-work}

\paragraph{Learning to price when buyers do not know their values.}
Learning to price when buyers do not know their values requires new machinery beyond classic pricing algorithms and auction design.
Prior works propose different strategies for the seller when the buyers learn through various means. 
 One line of work studies bidding strategies for buyers who do not know their values in auction settings~\citep{feng2018learning,weed2016online,kandasamy23vcg}.
Another line of work considers selling repeatedly to a single buyer while the buyer is learning from their own experience at each round~\citep{papadimitriou2022complexity, ashlagi2016sequential,Chawla22:Simple}. 
However, a significant limitation in practice is that buyers on online platforms do not necessarily return repeatedly to buy the same item
and can only obtain feedback from previous buyers via reviews. In this paper, we study the seller's pricing strategy when the buyers can only learn from past reviews.

\citet{Ifrach19:Bayesian} consider a similar pricing problem for the seller when the buyers learn from reviews.
However, their model is limited to one buyer type, where the buyers' values for the item are i.i.d. random variables from a fixed distribution. In contrast, we study the setting where there are multiple buyer types. Moreover, the seller does not know the frequency of each type and the type of buyer who arrives at each round, which leads to crucial difficulties in our analysis.

\paragraph{Learning to price when buyers know their values.} \citet{Zhao20:Stochastic}  study a setting where the buyers know their values, but the seller does not know the distribution over buyers' values. Reviews give the seller more information about this distribution than purchase decisions alone would. \citet{Zhao20:Stochastic} present an algorithm that uses the (non-noisy) reviews to obtain a $\tilde{O}\left(T^{\nicefrac{1}{2}}\right)$ regret bound. In contrast, if the seller only observes purchase decisions and not reviews, \citet{kleinberg2003value} provide a $\Omega\left(T^{\nicefrac{2}{3}}\right)$ lower bound.
While they show that this bound can be improved to $\tilde{\Theta}\left(T^{\nicefrac{1}{2}}\right)$,
it requires additional distributional assumptions.

\paragraph{Selling to no-regret buyers who know their values.} 
In situations where buyers know their values, 
the buyer may strategically improve their purchase decisions or bidding strategy over repeated interactions to achieve a higher accumulated utility. No-regret learning has been explored as a model of buyer behavior~\citep{braverman2018selling,deng2019prior,nekipelov2015econometrics,devanur2014perfect}. In this literature, buyers know their values but may use no-regret algorithms to learn how to bid. In comparison, in this paper, we work with buyers who do not know their values and need to estimate them from historical reviews.
This leads to different dynamics. For example, suppose a seller repeatedly sets the Myerson reserve
price. In that case, any buyer who knows her value \emph{a priori} and uses a no-regret algorithm will eventually
learn to submit a winning bid. However, a buyer without a reasonable estimate of her value
may consider the Myerson price too high and will not buy. Interestingly, a seller 
dealing with either type of learner may benefit from selling the item for a low price early on, but for two
very different reasons. 
In our setting,  this will give buyers of a given type the
opportunity to refine their estimated value and will encourage future buyers of the
same type to buy at higher prices if their value
is indeed high. On the other hand, as \citet{braverman2018selling} show, giving items for free
to agents who are learning to bid will accrue welfare (as long as agents are allowed to overbid), which the algorithm can then extract in future rounds by setting prices that are higher than agent values.

\paragraph{Buyers' social learning from reviews.}
Our work is also related to a rich literature on
buyer behavior and social learning from reviews when buyers do not know their values~\citep{Ifrach19:Bayesian, Boursier22:Social, han2020customer, chamley2004rational, besbes2018information, bose2006dynamic,Crapis17:Monopoly, kakhbod2021heterogeneous, Acemoglu22:Learning}. Much of the research on social learning from reviews can be categorized into two groups depending on whether the decision model is Bayesian or non-Bayesian. In the Bayesian model, \citet{Ifrach19:Bayesian}, \citet{Acemoglu22:Learning}, and \citet{Boursier22:Social} study a setting where the buyers decide whether to purchase the item by calculating posterior probabilities about the item's quality given the past reviews.

It may be computationally challenging for buyers to compute Bayesian updates, so several papers relax this assumption~\citep{Crapis17:Monopoly,besbes2018information}. \citet{besbes2018information}, for example, study both fully rational Bayesian buyers and buyers with limited rationality who can only observe the average of the past reviews. Under these two extremes, they analyze the conditions under which buyers can recover a product's true quality based on their observed feedback. Unlike our paper, the buyers have private signals about the item for sale, influencing their purchase decisions. Our model can be seen as situated between these two extremes because the purchase decisions depend on the average of the past reviews and the number of those reviews. Moreover, whereas \citet{besbes2018information} analyze risk-neutral buyers, we study a form of risk aversion where buyers may not purchase even if the price is below the average reviews.

Unlike this prior research, we do not assume all buyers share a specific decision policy. Instead, we identify a broad family of decision policies under which our results hold. In particular, we only require that the buyer purchases the item if the price is sufficiently low.

\section{Notation and online learning setup}

In our model, an item is sold repeatedly to a sequence of distinct buyers over a series of $T$ rounds. Each buyer has a type $i \in [d]$, and there is an unknown distribution $\cP$ over the types $[d]$.
We use the notation $q_i = \Pr_{j \sim \cP}[j = i]$ and $q_{\min} = \min_{i \in [d]}q_i.$

The \emph{ex-ante} value of a buyer with type $i \in [d]$ is $\theta_i \in [0,1]$.
If a buyer with type $i \in [d]$ purchases the item, their \emph{ex-post} value is drawn from a distribution $\cD_i$ with support $[0,1]$ and mean $\theta_i$.
The seller knows $\theta_1, \dots, \theta_d$ but not the distributions $\cP, \cD_1, \dots, \cD_d.$ For ease of analysis, we assume that the seller has ordered the types such that $\theta_1 \leq \theta_2 \leq \dots \leq \theta_d$, but the buyers are unaware of this ordering. This assumption is not necessary for the results to hold.

At each timestep $t \in [T]$:
\begin{enumerate}
    \item There is a set $\sigma_{t-1}$ of reviews which describe past buyers' types and their \emph{ex-post} values. 
    \item The seller first sets a price $p_t \in [0,1]$.
    \item A buyer arrives with type $i_t \sim \cP$. They observe the past reviews of buyers with type $i_t$: $\Phi_{i_t,t} = \left\{v : (i,v) \in \sigma_{t-1} \text{ and } i = i_t\right\}$. They decide whether to purchase the item using $\Phi_{i_t,t}$. We describe the buyer's purchasing model in more detail in Section~\ref{sec:purchasing}.
    Observe that the seller is unaware of the buyer's type $i_t$ when they set the price.
    \item If the buyer purchases the item, they pay $p_t$ and leave a review of $\left(i_t, v_t\right)$ describing both their type and their \emph{ex-post} value $v_t \sim \cD_{i_t}.$ In this case, $\sigma_t = \sigma_{t-1} \cup \left\{\left(i_t, v_t\right)\right\}$, and otherwise, $\sigma_t = \sigma_{t-1}$.
\end{enumerate}

Our assumptions and model reflect practical e-commerce settings.
First, quite often, it is reasonable to assume that sellers
know customers' \emph{ex-ante} values as they may have inside information.
For instance, a skincare product vendor may know that a particular product works better on some skin types.
However, buyers may not simply trust the seller if they were to publish this value, as the seller has every incentive to overstate this value to maximize revenue.
A buyer would instead decide if a product is suitable for her
via independent reviews from other customers.
Second, for fairness reasons, in e-commerce platforms, sellers typically have to publish a single price for all customers and cannot sell the item at individualized prices.
Third, if a buyer does not purchase an item, they will not leave a review, and the seller has no way of knowing their type or \emph{ex-post} value.

\subsection{Buyers' purchasing model}\label{sec:purchasing}

At time step $t$, the agent's purchase decision is defined by a threshold $\tau_t(\sigma_{t-1}, i_t) \geq 0$ that takes as input their type $i_t$ and the reviews left by past agents. Intuitively, $\tau_t(\sigma_{t-1}, i_t)$ represents the agent's estimate of their value $\theta_{i_t}$ based on past reviews. The agent purchases the item if $p_t \leq \tau_t(\sigma_{t-1}, i_t)$.

A conservative agent would choose $\tau_t(\sigma_{t-1}, i_t)$ to be low in order to always guarantee that $\tau_t(\sigma_{t-1}, i_t) \leq \theta_{i_t}$, so that they only purchase when their \emph{ex-ante} utility is non-negative. An extreme example of this type of conservatism would set $\tau_t(\sigma_{t-1}, i_t)= 0$, meaning that the agent would only purchase the item if offered for free. Optimizing revenue with such a conservative agent would be hopeless. Therefore, we impose the following natural lower bound on $\tau_t(\sigma_{t-1}, i_t)$:

\begin{definition}
\label{defn:deltapessimism}
    Let $\Phi_{t} \subseteq \sigma_{t-1}$ be the reviews left by agents with type $i_t$: \[\Phi_{t} = \left\{v : (i, v) \in \sigma_{t-1} \text{ and } i = i_t\right\}.\] Let $\text{LB}_{t}$ be the average of these reviews minus a standard confidence term: \[\text{LB}_{t}=\begin{cases} 0 &\text{if } \Phi_{t} = \emptyset,\\
\max\left\{\,0,\;\frac{1}{|\Phi_{t}|}\sum_{v \in \Phi_t}v - \sqrt{\frac{1}{2|\Phi_{t}|}\ln \frac{t}{\eta}}\right\} &\text{else.}\end{cases}\]
We say that the agent on round $t$ is \emph{$\eta$-pessimistic} if, $\tau_t(\sigma_{t-1}, i_t) \geq \text{LB}_t.$
\end{definition}

This uncertainty term corresponds to the standard confidence interval defined by the Hoeffding bound.  Intuitively, as a buyer sees more reviews from his type, this uncertainty decreases, and he is more certain about his \emph{ex-ante} valuation. The $\ln t$ term is necessary to construct a valid confidence interval for an arbitrary algorithm as the data may not be independent (see Appendix~\ref{app:LCB-def-proof}): the algorithm's price may depend on previous reviews, which in turn will affect future buyers and reviews. This $\ln t$ term is not fundamental—the lower bound does not use it.

Intuitively, the agents can be confident that \emph{regardless} of the policy used by the seller, with probability $1-\eta$, for all rounds $t \in [T]$, $\theta_{i_t} \geq \LB_t$.
We prove this formally in Appendix~\ref{app:LCB-def-proof}.
Therefore, if the price is lower than $\text{LB}_t$,
an $\eta$-pessimistic agent will buy the item as they can be confident, based on past reviews, that their \emph{ex-ante} utility $\theta_{i_t} - p_t$ will be non-negative. This restriction bounds the level of pessimism that the agents can display and thus makes it possible to set reasonable prices.
We clip this lower confidence bound at $0$ since valuations are always in $[0,1]$.

\subsection{Regret}

We define
regret as the difference between:\begin{enumerate}
    \item The algorithm's total expected revenue, and
    \item \emph{(baseline)} The expected revenue of the optimal fixed price if the agents bought whenever their \emph{ex-ante} value was larger than the price.
\end{enumerate}
Under the baseline that we compete with, both the buyer and the seller are equipped with more information than in the learning problem: the seller knows all distributions $\cP, \cD_1, \dots, \cD_d$ and the buyers know their \emph{ex-ante} values $\theta_1, \dots, \theta_d.$
Therefore, the seller knows \emph{a priori} which customers to target to maximize revenue.
Moreover,
since the buyers do not need to learn their \emph{ex-ante} values from reviews, the seller can extract higher revenue than they could from 
uncertain buyers who may only buy when the price is likely lower than their \emph{ex-ante} value.

Formally, let $b_t \in \{0,1\}$ indicate whether or not the buyer bought on round $t \in [T]$ and let $p^* = \argmax_{p \in [0,1]} p\Pr_{i \sim \cP}[\theta_i \geq p]$ be the price with highest expected revenue if the agents bought whenever their \emph{ex-ante} value was larger than the price. Regret is defined as 
\begin{align}
    \E\left[R_T\right] = Tp^*\Pr_{i \sim \cP}[\theta_i \geq p^*] - \E\left[\sum_{t = 1}^T p_tb_t\right].
    \label{eqn:regretdefn}
\end{align}

\section{Online Pricing Algorithm}

This section describes our algorithm, which has two phases: Algorithm~\ref{alg:type_elim} and \ref{alg:main}.
It is defined by a parameter $\lambda > 0$. (We will choose
$\lambda = d^{-2/3}T^{-1/3}$ to obtain optimal trade-offs).

\begin{algorithm}[t]
	\SetAlgoNoLine
	\KwIn{Number of timesteps $t_{\lambda}$, number of types $d$, parameter $\lambda \in [0,1]$}
	\For{$t = 1, \dots, t_{\lambda}$}{
		Set $p_t = 0$\\
		Buyer with type $i_t$ arrives and purchases item\\
        Buyer leaves review $(i_t, v_t)$, where $v_t\sim \cD_{i_t}$
	}
    \For{$i \in [d]$}{
        Set \[\overline q_i = \frac{1}{t_{\lambda}}\sum_{t=1}^{t_{\lambda}} \Ind(i_t = i)\]\Comment{Calculate the fraction of rounds each type appeared}
    }
    Set $Q = \left\{i : \overline{q}_i \geq \frac{3\lambda}{4}\right\}$ \Comment{Set of types that appeared at least a $(3\lambda/4)$-fraction of rounds}\\
    \KwOut{$Q$}
\caption{\textsc{TypeElimination}}
\label{alg:type_elim}
\end{algorithm}

\begin{algorithm}[t]
		\KwIn{Number of timesteps $T$, number of types $d$, $\eta \in [0,1]$ such that the agents are $\eta$-pessimistic, parameter $\lambda \in [0,1]$}
			Set $t_{\lambda} \defeq \frac{32\ln (dT^2)}{\lambda} + 1$\\
	Set $S_{t_{\lambda}+1} = \textsc{TypeElimination}(t_{\lambda}, d, \lambda)$ \Comment{$S_t$ is the set of ``active types''}\\
	\For{$t = t_{\lambda}+1, \dots, T$}{
		\For{$i \in S_t$}{
			Compute $\Phi_{it} = \left\{v_s : (i,v_s) \in \sigma_{t-1}\right\}$ and
			\[\text{LB}_{it}=\begin{cases} 0 &\text{if } \Phi_{it} = \emptyset\\
				\max\left\{\frac{1}{|\Phi_{it}|}\sum_{v \in \Phi_{it}}v - \sqrt{\frac{1}{2|\Phi_{it}|}\ln \frac{T}{\eta}}, 0\right\} &\text{else}\end{cases}\]
		}
		Set price $p_t = \min_{i \in S_t}\left\{\min\left\{\theta_{i}, \LB_{it}\right\}\right\}$\Comment{$p_t$ is the smallest $\LB_{it}$ or $\theta_i$ of any active type}\\
		$b_t = \Ind(\text{buyer buys at price $p_t$})$ \Comment{We prove that if $i_t \in S_t$, then $b_t = 1$}
		
		\If{$b_t = 1$}{
			The buyer leaves a review $(i_t, v_t)$ where $v_t \sim \cD_{i_t}$
		}
		Set $\rho_t = \sqrt{\frac{\ln (dT^2)}{2(t - t_{\lambda})}}$\\
		
		\For {$i \in S_t$}{
				$\overline{\mu}_{i,t} = \frac{1}{t-t_{\lambda}}\sum_{s = t_{\lambda} + 1}^t \theta_{i} \cdot \Ind(b_s = 1 \land \theta_{i_s} \geq \theta_{i} \land i_s \in Q)$ \Comment{Estimate of $\rev\left(\theta_i, Q\right)$}\\
				$\widehat{\mu}_{i,t}= \overline{\mu}_{i, t} + \rho_t$ \Comment{Upper confidence bound}\\
				$\widecheck{\mu}_{i,t} = \overline{\mu}_{i,t} - \rho_t$\Comment{Lower confidence bound}
		}
		Set $i_0 = \min\left\{i \in S_t : \widehat{\mu}_{i,t} \geq \max_{k \in S_t} \widecheck{\mu}_{k, t}\right\}$ \Comment{For $i < i_0$, $\rev\left(\theta_i, Q\right)$ is likely too small}\\
		Set $S_{t+1} = S_t \cap \{i_0, i_0 + 1, \dots, d\}$ \Comment{Eliminate types $i < i_0$}
	}
	\caption{Online pricing with reviews}
	\label{alg:main}
\end{algorithm}

Our algorithm has two phases.
In the first phase (Algorithm~\ref{alg:type_elim}), the algorithm sets a price of 0 for $t_{\lambda} = \Theta(\ln(dT)/\lambda)$ rounds. The agent will buy the item at each round since the price is 0 and leave a review. This allows the algorithm to obtain i.i.d. samples from the type distribution $\cP$.
In phase 2 (Algorithm~\ref{alg:main}), i.e, the remaining $T - t_{\lambda}$ rounds, the algorithm will ignore types that appeared too rarely during phase 1---in particular, on fewer than a $(3\lambda/4)$-fraction of rounds. Intuitively, customers of these types have a low probability of appearance and thus will have more uncertainty about their values due to fewer reviews.
The uncertainty term will cause the lower confidence bound $\LB_t$ in Definition~\ref{defn:deltapessimism} to be small.
As the seller will have to choose a low price to target these customers (even if their ex-ante value is large), they may have to forego higher revenue from more frequent customer types.
Therefore, it is not worthwhile for the algorithm to target these customers. We use $Q$ to denote the buyer types that appeared on at least a $(3\lambda/4)$-fraction of rounds.

To describe the algorithm's second phase, we will use the notation \[\rev(p,Q) = p \Pr_{i \sim \cP}\left[\theta_i \geq p \text{ and }i \in Q\right],\] to denote the expected revenue of a price $p$ restricted to buyers in $Q$ and $p^\ast(Q) = \argmax \rev(p,Q)$.
In this phase, Algorithm~\ref{alg:main} will ignore the extremely rare buyers not in $Q$ and aim to set prices that compete with $p^*(Q)$. In the analysis, we will show that by competing with $p^*(Q)$, Algorithm~\ref{alg:main} also competes with the optimal price $p^*$.
 
Observe that $p^*(Q) = \theta_{i_Q}$ for some $i_Q \in Q$.
On each round $t > t_{\lambda}$ of the second phase, Algorithm~\ref{alg:main} maintains a set $S_t$ of ``active types'' such that $i_Q$ is likely in $S_t.$ Algorithm~\ref{alg:main} sets the price $p_t$ low enough to ensure that if the current type $i_t$ is in $S_t$, then the buyer will buy. In particular, we define $\LB_{it}$ as the largest price the seller can set to ensure a purchase from a buyer of type $i$. We then set the price $p_t$ to be the smallest $\LB_{i,t}$ or $\theta_i$ of any active type $i \in S_t$ (we include $\theta_i$ for ease of analysis). If the buyer purchases the item, they leave a review $(i_t, v_t)$ where $v_t \sim \cD_{i_t}.$ 
 
Next, for each active type $i \in S_t$, the seller estimates $\rev(\theta_i, Q)$. We denote this estimate as $\overline{\mu}_{i,t}$ along with upper and lower confidence bounds $\widehat{\mu}_{i,t}$ and $\widecheck{\mu}_{i,t}$. We will describe this estimate more in Section~\ref{sec:upperbound}.
When estimating the revenue for different prices via the averages $\overline{\mu}_{i,t}$, we
only use samples from the second phase. 
Doing so leads to a cleaner analysis, allowing us to separate the randomness in eliminating low probability types to determine the set $Q$ from the randomness of estimating $\rev(\theta_i, Q)$.
However, when constructing the lower confidence bound ${\rm LB}_{i,t}$ for
customers of type $i$, we use reviews from all rounds.
This is to be expected, as customers will use all past reviews when making a purchasing decision.

Algorithm~\ref{alg:main} defines \[i_0 = \min\left\{i \in S_t : \widehat{\mu}_{i,t} \geq \max_{k \in S_t} \widecheck{\mu}_{k, t}\right\}\] to be the smallest active type such that $\theta_{i_0}$ may plausibly be $p^*(Q)$. For all $i < i_0$, the upper confidence bound on $\rev\left(\theta_i, Q\right)$ is small $\left(\widehat{\mu}_{i,t} < \max_{k \in S_t} \widecheck{\mu}_{k, t}\right)$, so it is unlikely that $\theta_i = p^*(Q).$ Algorithm~\ref{alg:main} concludes round $t$ by eliminating all types $i < i_0$ from the active set.

\section{Regret upper bounds}
\label{sec:upperbound}

We now state our main upper bounds on regret (Equation~\eqref{eqn:regretdefn}).

\begin{theorem}\label{thm:regret}
    Suppose the agents are $\eta$-pessimistic. If $q_{\min} \leq 2\lambda$ then \[\E[R_T] = O\left(\frac{\ln(dT)}{\lambda} + Td\lambda + \sqrt{T\ln (dT)} + \sqrt{\frac{T}{\lambda}\ln \frac{dT}{\eta}}\right),\] and if $q_{\min} > 2\lambda$, then \[\E[R_T] = O\left(\frac{\ln(dT)}{\lambda} + \sqrt{T \ln(dT)} + \sqrt{\frac{T}{q_{\min}}\ln \frac{dT}{\eta}}\right).\]
\end{theorem}

Theorem~\ref{thm:regret} implies the following corollary for the specific choice of $\lambda = d^{-2/3}T^{-1/3}$.

\begin{corollary}
\label{cor:regret}
Suppose the agents are $\eta$-pessimistic. Setting $\lambda = d^{-2/3}T^{-1/3}$, we have that 
    if $q_{\min} \leq 2d^{-2/3}T^{-1/3}$ then \[\E[R_T] = O\left(T^{2/3}d^{1/3} + T^{1/3}d^{1/3}\sqrt{\ln \frac{dT}{\eta}}  + \sqrt{T\ln (dT)}+ T^{1/3}d^{2/3}\ln(dT)\right),\] and if $q_{\min} > 2d^{-2/3}T^{-1/3}$, then \[\E[R_T] = O\left(\sqrt{\frac{T}{q_{\min}}\ln \frac{dT}{\eta}} + T^{1/3}d^{2/3}\ln(dT)\right).\]
\end{corollary}

We note that while the worst-case regret scales with $T^{2/3}$, it improves to $\sqrt{T}$ when all types appear with large enough probability since customers of all types will be able to form accurate estimates of their values quickly.
We emphasize that our algorithm and analysis are markedly different from explore-then-commit (ETC) style algorithms in stochastic bandit settings, which share a similar two-phase strategy and have
$T^{2/3}$ regret.
First, the first `explore' phase of ETC algorithms is much longer
(typically $\tilde O(T^{2/3})$ rounds) than our Phase 1, which lasts only  $\tilde O(T^{1/3})$ rounds. ETC algorithms also focus on learning all unknowns in their first phase, while here, its only purpose is to eliminate low probability types.
Second, in the `commit' phase of ETC algorithms, typically, no learning is required, while in our second phase, the algorithm is still learning the optimal price.
Third, unlike our algorithm, ETC algorithms cannot obtain $\sqrt{T}$ regret even under favorable conditions~\citep{garivier2016explore}.
Fourth, we reiterate that the $T^{2/3}$ worst-case regret is due to the uncertainty on the buyers' side, which is a challenge specific to our setting.

We will first provide an overview of our proof, with the full proof to
follow in Section~\ref{sec:upper_proof}.

\begin{proof}[Proof sketch of Theorem~\ref{thm:regret}]
The terms of our regret bounds in Theorem~\ref{thm:regret} arise from the following steps of our analysis.
The first phase immediately contributes $O\left(\frac{\ln (dT)}{\lambda}\right)$ to the regret since the item is sold for free during that phase.
At the end of the first phase, Algorithm~\ref{alg:type_elim} discards the types that appeared too infrequently, resulting in a set $Q$, and only aims to maximize revenue over $Q$. When $q_{\min} \leq 2\lambda$, we prove that competing with $p^*(Q)$ rather than $p^*([d])$ contributes $Td\lambda$ to the regret. Meanwhile, when $q_{\min} > 2\lambda$, we prove that with high probability, $Q = [d]$, and thus $p^*(Q) = p^*([d])$, so there is no impact on regret.

In order to gradually learn a price that competes with $p^*(Q)$, Algorithm~\ref{alg:main} maintains estimates $\overline{\mu}_{i,t}$ of $\rev(\theta_i, Q) = \theta_i \Pr_{j \sim \cP}\left[\theta_j \geq \theta_i \text{ and } j \in Q\right]$ for the active types $i \in S_t$. The error of these estimates contributes a factor of $O\left(\sqrt{T\ln (dT)}\right)$ to the regret. This step of the analysis takes some care because we cannot observe at each round $t$ whether or not $i_t \in Q$, provided the buyer did not buy the item.
If we were able to observe whether $i_t \in Q$, we could simply set \[\overline{\mu}_{i,t} = \frac{1}{t-t_{\lambda}}\sum_{s = t_{\lambda} + 1}^t \theta_{i} \cdot \Ind(\theta_{i_s} \geq \theta_{i} \text{ and } i_s \in Q),\] and the concentration would follow from a Hoeffding bound. Instead, we set \[\overline{\mu}_{i,t} = \frac{1}{t-t_{\lambda}}\sum_{s = t_{\lambda} + 1}^t \theta_{i} \cdot \Ind(b_s = 1 \text{, } \theta_{i_s} \geq \theta_{i} \text{, and } i_s \in Q),\] but nonetheless prove that it is a good estimate of $\rev\left(\theta_i, Q\right)$. To do so, we show that for all active types $i \in S_t$ and all rounds $s \leq t$ of Algorithm~\ref{alg:main}, \begin{equation}\Ind(\theta_{i_s} \geq \theta_{i} \text{ and } i_s \in Q) = \Ind(b_s = 1 \text{, } \theta_{i_s} \geq \theta_{i} \text{, and } i_s \in Q),\label{eq:estimate_proof_sketch}\end{equation} so we can still apply a Hoeffding bound (taking into account that the set $S_t$ is a random variable). If $b_s = 1$, then clearly Equation~\eqref{eq:estimate_proof_sketch} holds. Otherwise, $i_s \not\in S_s$ because any buyer in $S_s$ will always buy. We show that this means that either $i_s \not\in Q$ or---based on the way that types are eliminated from the active sets---$\theta_{i_s} < \theta_i$, so Equation~\eqref{eq:estimate_proof_sketch} holds in this case as well.

Finally, the agents themselves are learning as the algorithm progresses, which increases the regret since our benchmark is the expected revenue of the optimal price if the agents buy whenever their \emph{ex-ante} value is larger than the price. When $q_{\min} \leq 2\lambda$, the fact that the agents are learning contributes $O\left(\sqrt{\frac{T}{\lambda}\ln \frac{dT}{\eta}}\right)$ to the regret and when $q_{\min} > 2\lambda,$ it contributes $O\left(\sqrt{\frac{T}{q_{\min}}\ln \frac{dT}{\eta}}\right)$.
\end{proof}

\subsection{Proof of the regret upper bound (Theorem~\ref{thm:regret})}\label{sec:upper_proof}

In this section, we prove Theorem~\ref{thm:regret}.
The proof relies on a handful of helper lemmas which we prove in Appendix~\ref{app:lemmas}.

\begin{proof}[Proof of Theorem~\ref{thm:regret}]
 In this proof, on each round $t>t_{\lambda}$, we use the notation $p_t' = \min_{i\in S_t}\theta_{i}$. We split the regret into five terms as follows:  \[R_T = \sum_{t = 1}^T p^*([d]) \Ind\left(\theta_{i_t} \geq p^*([d])\right) - \sum_{t = 1}^T p_tb_t = Z_1 + Z_2 + Z_3 + Z_4 + Z_5.\] The first term $Z_1 = p^*([d]) t_{\lambda} \leq \frac{32\ln (dT^2)}{\lambda} + 1$ measures the revenue lost from offering the item for free for the first $t_{\lambda}$ rounds. The second term \begin{align*}Z_2 &= \sum_{t > t_{\lambda}}p^*([d])\Ind\left(\theta_{i_t} \geq p^*([d])\right) - \sum_{t > t_{\lambda}}p^*([d])\Ind\left(\theta_{i_t} \geq p^*([d]) \text{ and } i_t \in Q\right)\\
        &= \sum_{t > t_{\lambda}}p^*([d])\Ind\left(\theta_{i_t} \geq p^*([d]) \text{ and } i_t \not\in Q\right)\end{align*} relates to the revenue lost due to the fact that we only aim to compete with the optimal price for relatively-common types---namely those in $Q$---as does the third term \[Z_3 = \sum_{t > t_{\lambda}}p^*([d])\Ind\left(\theta_{i_t} \geq p^*([d]) \text{ and } i_t \in Q\right) - \sum_{t > t_{\lambda}}p^*(Q)\Ind\left(\theta_{i_t} \geq p^*(Q) \text{ and } i_t \in Q\right).\]
        The fourth term \[Z_4 = \sum_{t > t_{\lambda}}p^*(Q)\Ind\left(\theta_{i_t} \geq p^*(Q) \text{ and } i_t \in Q\right) - \sum_{t > t_{\lambda}}p_{t}'\Ind\left(\theta_{i_t} \geq p_{t}' \text{ and } i_t \in Q\right)\] relates the cumulative revenue of the optimal price over $Q$---that is, $p^*(Q)$---to the cumulative revenue of the ``proxy'' price $p_t' = \min_{i\in S_t}\theta_{i}$.
        Finally, the last term \[Z_5 = \sum_{t > t_{\lambda}}p_{t}'\Ind\left(\theta_{i_t} \geq p_{t}' \text{ and } i_t \in Q\right)- \sum_{t >t_{\lambda}}p_tb_t\] relates the cumulative revenue of the proxy price $p_t'$ to the algorithm's cumulative revenue.
    In the following claims, we bound $Z_2, Z_3, Z_4,$ and $Z_5.$ The full proofs are in Appendix~\ref{app:upper_proof}.

    \begin{restatable}{claim}{Ztwo}\label{claim:Z2}
    If $q_{\min} \leq 2\lambda$ then $\E[Z_2] \leq Td\lambda + 1$ and if $q_{\min} > 2\lambda$, then $\E[Z_2] \leq 1.$
    \end{restatable}
    \begin{proof}[Proof sketch of Claim~\ref{claim:Z2}]
    First, we bound $Z_2$ as follows: \[Z_2 = \sum_{t > t_{\lambda}}p^*([d])\Ind\left(\theta_{i_t} \geq p^*([d]) \text{ and } i_t \not\in Q\right) \leq \sum_{t > t_{\lambda}}p^*([d])\Ind\left(i_t \not\in Q\right) \leq \sum_{t > t_{\lambda}}\Ind\left(i_t \not\in Q\right).\]
    Recall from Algorithm~\ref{alg:type_elim} that $\overline{q}_i$ is the fraction of times that type $i$ appears in phase 1 and let $\cG$ be the event that for all $i \in [d]$ such that $q_i \geq \lambda$, we have that $\overline{q}_i \geq \frac{3\lambda}{4}$, which means that $i \in Q$. In other words, when $\cG$ happens, $[d] \setminus Q \subseteq \{q_i : q_i < \lambda\}$. In Lemma~\ref{lem:G}, we prove that $\Pr[\cG^c] \leq \frac{1}{T},$ so $\E[Z_2] \leq \E[Z_2 \mid \cG] + T\Pr[\cG^c] \leq \E[Z_2 \mid \cG] + 1$.

    Next, since $Q$ is a random variable, we condition on it as well: \[\E[Z_2 \mid \cG] = \sum_{Q' \subseteq [d]} \E[Z_2 \mid Q = Q',  \cG]\Pr[Q = Q' \mid \cG].\] If $[d]\setminus Q' \not\subseteq \left\{q_i : q_i < \lambda\right\}$, then $\Pr[Q = Q' \mid \cG] = 0$. For any $Q'$ such that $[d]\setminus Q' \subseteq \left\{q_i : q_i < \lambda\right\}$, we prove \[\E[Z_2 \mid Q = Q',  \cG] = \sum_{t > t_{\lambda}}\sum_{i \not\in Q'}\Pr\left[i_t = i\right].\]

If $q_{\min} \leq 2\lambda$, then \[\sum_{t > t_{\lambda}}\sum_{i \not\in Q'}\Pr\left[i_t = i\right] \leq \sum_{t > t_{\lambda}}\sum_{i \not\in Q'}\lambda \leq Td\lambda,\] which implies that $\E[Z_2] \leq Td\lambda + 1.$
The case where $q_{\min} > 2\lambda$ follows similarly.
        \end{proof}

\begin{restatable}{claim}{Zthree}\label{claim:Z3}
    $\E[Z_3] \leq 0.$
\end{restatable}
\begin{proof}[Proof sketch of Claim~\ref{claim:Z3}]
In this proof we bound \begin{equation}Z_3 = \sum_{t > t_{\lambda}}p^*([d])\Ind\left(\theta_{i_t} \geq p^*([d]) \text{ and } i_t \in Q\right) - \sum_{t > t_{\lambda}}p^*(Q)\Ind\left(\theta_{i_t} \geq p^*(Q) \text{ and } i_t \in Q\right).\label{eq:Z3}\end{equation}
We begin by conditioning the first term of Equation~\eqref{eq:Z3} on $Q$ since it is a random variable: \begin{align}&\E\left[\sum_{t > t_{\lambda}}p^*([d])\Ind\left(\theta_{i_t} \geq p^*([d]) \text{ and } i_t \in Q\right)\right]\nonumber\\
            =&\sum_{Q' \subseteq [d]}\sum_{t > t_{\lambda}}p^*([d])\Pr\left[\theta_{i_t} \geq p^*([d]) \text{ and } i_t \in Q' \mid Q = Q'\right]\Pr[Q = Q']\nonumber\\
            \leq \, &\sum_{Q' \subseteq [d]}\left(T-t_{\lambda}\right)p^*(Q')\Pr_{i \sim \cP}\left[\theta_i \geq p^*(Q') \text{ and } i \in Q' \right]\Pr[Q = Q']\label{eq:first_term},\end{align} where the final inequality follows from the definition of $p^*(Q')$.

Next, for the second term of Equation~\eqref{eq:Z3}, 
        \begin{align*}&\E\left[\sum_{t > t_{\lambda}}p^*(Q)\Ind\left(\theta_{i_t} \geq p^*(Q) \text{ and } i_t \in Q\right)\right]\\
            =&\sum_{Q' \subseteq [d]}\E\left[\sum_{t > t_{\lambda}}p^*(Q')\Ind\left(\theta_{i_t} \geq p^*(Q') \text{ and } i_t \in Q'\right) \mid Q = Q'\right]\Pr[Q = Q']\\
            = \, &\sum_{Q' \subseteq [d]}\left(T - t_{\lambda}\right)p^*(Q')\Pr_{i \sim \cP}\left[\theta_i \geq p^*(Q') \text{ and } i \in Q'\right]\Pr[Q = Q'].\end{align*}
Combined with Equation~\eqref{eq:first_term}, we have that $\E[Z_3] \leq 0$.
\end{proof}

\begin{restatable}{claim}{Zfour}\label{claim:Z4}
    $\E[Z_4] \leq 5 + 4\sqrt{2T\ln(dT^2)}$.
\end{restatable}

\begin{proof}[Proof sketch of Claim~\ref{claim:Z4}]
In this claim, we bound \begin{equation}Z_4 = \sum_{t > t_{\lambda}}p^*(Q)\Ind\left(\theta_{i_t} \geq p^*(Q) \text{ and } i_t \in Q\right) - \sum_{t > t_{\lambda}}p_{t}'\Ind\left(\theta_{i_t} \geq p_{t}' \text{ and } i_t \in Q\right)\label{eq:Z4}\end{equation} where $p_t' = \min_{i\in S_t}\theta_{i}$. Beginning with the first term of this equation, we prove that \begin{equation}\sum_{t > t_{\lambda}}\E\left[
p^*(Q)\Ind\left(\theta_{i_t} \geq p^*(Q) \text{ and } i_t \in Q\right)\right] = \sum_{t > t_{\lambda}}\E\left[\rev(p^*(Q), Q)\right]\label{eq:Z4_first}.\end{equation}

Moving on to the second term of Equation~\eqref{eq:Z4}, we prove that for any $t > t_{\lambda}$, 
\begin{equation}\E\left[
p_{t}'\Ind\left(\theta_{i_t} \geq p_{t}' \text{ and } i_t \in Q\right)\right]= \E\left[\rev\left(
p_t', Q\right)\right].\label{eq:Z4_second}\end{equation}
Combining Equations~\eqref{eq:Z4_first} and \eqref{eq:Z4_second}, we have that \begin{equation}\E[Z_4] \leq \sum_{t > t_{\lambda}}\E\left[\rev\left(
p^*(Q), Q\right) - \rev\left(
p_t', Q\right)\right].\label{eq:Z4_bound}\end{equation}

Next, for all $t > t_{\lambda}$, let $\cB_t$ be the event that:
    \begin{enumerate}
    \item $i_Q \in S_{t}$ and
    \item $\rev(p^*(Q), Q) - \rev\left(p_t', Q\right) \leq 4\rho_{t-1}$ (where $\rho_{t_{\lambda}} = 1$).
    \end{enumerate}
    Also, let $\cC_t = \bigcap_{s = t_{\lambda} + 1}^t \cB_s.$ In Lemma~\ref{lemma:SE}, we prove that $\Pr\left[\cC_t^c\right] \leq \frac{1}{T}.$ By Equation~\eqref{eq:Z4_bound}, \[\E[Z_4] \leq \E\left[\left. \sum_{t > t_{\lambda}}\rev\left(
p^*(Q), Q\right) - \rev\left(
p_t', Q\right) \, \right| \, \cC_T\right] + T\Pr\left[\cC_T^c\right] \leq \sum_{t > t_{\lambda}} 4\rho_{t-1} + 1\]
which implies the result.    \end{proof}

    \begin{restatable}{claim}{Zfive}\label{claim:Z5}
    If $q_{\min} \leq 2\lambda$, then
        $\E[Z_5] \leq 4\sqrt{\frac{2T}{\lambda} \ln \frac{dT^2}{\eta}} + 3$ and
    if $q_{\min} > 2\lambda$, $\E[Z_5] \leq 4\sqrt{\frac{T}{q_{\min}} \ln \frac{dT^2}{\eta}} + 2.$
    \end{restatable}
    
    \begin{proof}[Proof sketch of Claim~\ref{claim:Z5}]
    On each round $t > t_{\lambda}$,
        recall that \[\text{LB}_{it}=\begin{cases} 0 &\text{if } \Phi_{it} = \emptyset\\
\max\left\{\frac{1}{|\Phi_{it}|}\sum_{v \in \Phi_{it}}v - \sqrt{\frac{1}{2|\Phi_{it}|}\ln \frac{T}{\eta}}, 0\right\} &\text{else.}\end{cases}\] Let $j_t = \argmin_{j \in S_t} \theta_{j}$, so $p_t' = \theta_{j_t}$. We prove that
        \[Z_5 \leq \sum_{t > t_{\lambda}}p_t'\Ind\left(\theta_{i_t} \geq p_t' \land i_t \in Q\land b_t = 0\right)+ \sum_{t > t_{\lambda}}(p_t' - p_t)b_t.\] Since $p_t \leq p_t'$, we have that \[Z_5 \leq \sum_{t > t_{\lambda}}p_t'\Ind\left(\theta_{i_t} \geq p_t' \land i_t \in Q\land b_t = 0\right)+ \sum_{t > t_{\lambda}}\left(p_t' - p_t\right).\] By definition of the pricing rule, if $i_t \in S_t$, then $b_t = 1$. Therefore, if $b_t = 0$, then either $i_t \not\in Q$ or $i_t \in Q \setminus S_t$. Since $S_t$ contains every $i \in Q$ with $i > j_t$, we can conclude that if $i_t \in Q \setminus S_t$, then $\theta_{i_t} < \theta_{j_t} = p_t'$. Therefore, $\Ind\left(\theta_{i_t} \geq p_t' \land i_t \in Q \land b_t = 0\right) = 0$, which means that \[\E\left[Z_5\right] \leq \E\left[\sum_{t> t_{\lambda}}p_t' - p_t\right].\]

Let $j'_t = \argmin_{j \in S_t} \LB_{jt}$, which means that
$p_t = \min_{i \in S_t}\left\{\min\left\{\theta_{i}, \LB_{it}\right\}\right\} = \min\left\{p_t', \LB_{j_t't}\right\}$. We also know that $p_t' = \theta_{j_t} \leq \theta_{j_t'}.$ Therefore, \[\E\left[Z_5\right] \leq \E\left[\sum_{t> t_{\lambda}}\max\left\{0, \theta_{j_t'} - \LB_{j_t't}\right\}\right].\]

Let $\cE_1$ be the event that for all $t > t_{\lambda}$, $\left|\Phi_{i,t}\right| \geq \frac{1}{2}q_{\min}(t-1)$ for all $i \in S_t$. In Lemma~\ref{lemma:reviews_large_q}, we prove that if $q_{\min} > 2\lambda$, then $\Pr\left[\cE_1^c\right] \leq \frac{1}{T}.$ Also, let $\cH$ be the event that for all $t > t_{\lambda}$ and all $i \in S_t$, \[|\Phi_{it}|\theta_{i} \leq \sum_{v \in \Phi_{it}} v + \sqrt{\frac{1}{2}|\Phi_{it}|\ln(dT^2)}.\] In Lemma~\ref{lem:H}, we prove that $\Pr[\cH^c] \leq \frac{1}{T}.$

Suppose that $q_{\min} > 2\lambda.$ In this case, \begin{align*}\E\left[Z_5\right] 
                &\leq \E\left[\left. \sum_{t> t_{\lambda}}\max\left\{0, \theta_{j_t'} - \LB_{j_t't}\right\} \, \right| \, \cE_1 \land \cH \right] + 2\\
                &= \E\left[\left. \sum_{t> t_{\lambda}}\max\left\{0, \theta_{j_t'} - \frac{1}{|\Phi_{j_t't}|}\sum_{v \in \Phi_{j_t't}}v + \sqrt{\frac{1}{2|\Phi_{j_t't}|}\ln \frac{1}{\eta}}\right\} \, \right| \, \cE_1 \land \cH \right] + 2.\end{align*}
Under events $\cE_1$ and $\cH$, \[\E[Z_5] \leq \E\left[\left. \sum_{t> t_{\lambda}} \sqrt{\frac{\ln(dT^2)}{2|\Phi_{j_t't}|}} + \sqrt{\frac{1}{2|\Phi_{j_t't}|}\ln \frac{1}{\eta}} \, \right| \, \cE_1 \land \cH \right] + 2\] and by definition of the event $\cE_1$,
\begin{align*}
                \E[Z_5]&\leq \E\left[\left. \sum_{t> t_{\lambda}} \sqrt{\frac{\ln(dT^2)}{2|\Phi_{j_t't}|}} + \sqrt{\frac{1}{2|\Phi_{j_t't}|}\ln \frac{1}{\eta}} \, \right| \, \cE_1 \land \cH \right] + 2\\
                &\leq \sum_{t = 2}^T \left(\sqrt{\frac{\ln(dT^2)}{q_{\min}(t-1)}} + \sqrt{\frac{1}{q_{\min}(t-1)}\ln \frac{1}{\eta}}\right) + 2 \leq 4\sqrt{\frac{T}{q_{\min}} \ln \frac{dT^2}{\eta}} + 2.
                \end{align*}

The proof when $q_{\min} < 2\lambda$ follows similarly.
\end{proof}
The final regret bound follows by combining these claims.
\end{proof}

We conclude this section by providing a proof sketch of one of the lemmas we used in Theorem~\ref{thm:regret}. The full proof and the remaining lemmas are in Appendix~\ref{app:upper_proof}.
This lemma shows that for all active types $i \in S_t$, $\overline{\mu}_{i,t}$ is indeed a good estimate of $\rev\left(\theta_i, Q\right) = \theta_i \Pr_{j \sim \cP}\left[\theta_j \geq \theta_i \text{ and } j \in Q\right]$.
As we described in the proof sketch of Theorem~\ref{thm:regret}, this takes some care because we cannot observe at each round $t$ whether or not $i_t \in Q$, provided the buyer did not buy the item.

\begin{restatable}{lemma}{A}\label{lem:A}
For all $t > t_{\lambda}$, let $\cA_t$ be the event $\rev\left(\theta_{i}, Q\right) \in \left[\widecheck{\mu}_{i,t} , \widehat{\mu}_{i,t} \right]$ for all $i \in S_t$. Then $\Pr[\cA_t^c] \leq \frac{1}{T^2}.$
\end{restatable}
\begin{proof}[Proof sketch]
Recall that $\widehat{\mu}_{i,t}= \overline{\mu}_{i,t} + \rho_t$ and $\widecheck{\mu}_{i,t} = \overline{\mu}_{i,t} - \rho_t$ with \[\rho_t = \sqrt{\frac{\ln (dT^2)}{2\left(t-t_{\lambda}\right)}}.\] We also define the related quantities for all $i \in [d]$ and all $Q' \subseteq [d]$: \[\overline{\gamma}_{i,t}(Q')
            = \frac{1}{t-t_{\lambda}}\sum_{s = t_{\lambda} + 1}^{t} \theta_i \cdot \Ind\left(\theta_{i_s} \geq \theta_i \land i_s \in Q'\right),\] $\widehat{\gamma}_{i,t}(Q')= \overline{\gamma}_{i,t}(Q') + \rho_t$, and $\widecheck{\gamma}_{i,t}(Q') = \overline{\gamma}_{i,t}(Q') -\rho_t$. By a Hoeffding bound, for all $Q' \subseteq [d]$ and $i \in [d]$, \[\Pr\left[\rev\left(\theta_{i}, Q'\right)  \not\in \left[\widecheck{\gamma}_{i,t}(Q'), \widehat{\gamma}_{i,t}(Q')\right]\right] \leq \frac{1}{dT^2}.\]

We claim that for any $i \in S_{t}$ and any $s > t_{\lambda}$, \begin{equation}\Ind(b_s = 1 \land \theta_{i_s} \geq \theta_{i} \land i_s \in Q) = \Ind\left(\theta_{i_s} \geq \theta_i \land i_s \in Q\right),\label{eq:estimators_equal_small_q}\end{equation} which means that $\overline{\mu}_{i, t} = \overline{\gamma}_{i, t}(Q),$ $\widehat{\mu}_{i, t} = \widehat{\gamma}_{i, t}(Q)$, and $\widecheck{\mu}_{i, t} = \widecheck{\gamma}_{i, t}(Q)$. To see why, if $b_s = 1$, then clearly Equation~\eqref{eq:estimators_equal_small_q} holds. Otherwise, suppose $b_s = 0$, in which case $\Ind(b_s = 1 \land \theta_{i_s} \geq \theta_{i} \land i_s \in Q) = 0$. Then $i_s \not \in S_s$ because any buyer in $S_s$ will always buy by the definition of the pricing rule. Let $j_s = \min\left\{j \in S_s\right\}.$ Since $S_s$ contains every element in $Q$ larger than $j_s$, we know that either:
            \begin{enumerate}
            \item $i_s \not\in Q$, in which case $\Ind\left(\theta_{i_s} \geq \theta_i \land i_s \in Q\right) = 0$, or 
            \item $i_s \in Q$ but $i_s \not\in S_s$, which means that $\theta_{i_s} < \theta_{j_s}$. Since $i \in S_t$, it must be that $i \in S_s$, so $\theta_{i_s} < \theta_{j_s} \leq \theta_{i}$. In this case, $\Ind\left(\theta_{i_s} \geq \theta_i \land i_s \in Q\right) = 0$ as well.
            \end{enumerate}
Therefore, Equation~\eqref{eq:estimators_equal_small_q} holds.

The fact that $\overline{\mu}_{i, t} = \overline{\gamma}_{i, t}(Q),$ $\widehat{\mu}_{i, t} = \widehat{\gamma}_{i, t}(Q)$, and $\widecheck{\mu}_{i, t} = \widecheck{\gamma}_{i, t}(Q)$ for all $i \in S_t$ implies that \begin{align}
    \Pr[\cA_t^c] &= \Pr\left(\exists i \in S_{t} \text{ s.t. }\rev\left(\theta_{i}, Q\right) \not\in \left[\widecheck{\mu}_{i,t} , \widehat{\mu}_{i,t} \right]\right)\nonumber\\
    &\leq \Pr\left(\exists i \in [d] \text{ s.t. }\rev\left(\theta_{i}, Q\right) \not\in \left[\widecheck{\gamma}_{i,t}(Q) , \widehat{\gamma}_{i,t}(Q) \right]\right)\nonumber\\
    &\leq \sum_{i = 1}^d\Pr\left(\rev\left(\theta_{i}, Q\right) \not\in \left[\widecheck{\gamma}_{i,t}(Q) , \widehat{\gamma}_{i,t}(Q) \right]\right).\label{eq:sum_d}
\end{align}
The result now follows from a union bound.
\end{proof}

\section{Regret lower bounds}

\label{sec:lowerbound}

In this section, we state our regret lower bounds.
Recall that
$\qmin = \min_{i\in [d]} \Pr_{j\sim\cP}[j=i]$ denotes the minimum probability of appearance among all
types.
Let $\RT(A, P)$ denote the regret after $T$ rounds when
using an algorithm $A$ on a problem $P$.
Our theorem below presents two lower bounds that correspond to our upper bounds.
First, we prove a $\qmin$ independent $\widetilde{\Omega}\left(T^{\nicefrac{2}{3}}
d^{\nicefrac{1}{3}}\right)$ lower bound on the regret.
Next, when $\qmin$ is large, we show that
$\widetilde{\Omega}\left(\sqrt{T/\qmin}\right)$ regret is still unavoidable.

\begin{theorem}\label{thm:lb}
For $T \in \Omega\left(d\left(\ln (1/\eta)\right)^2(\ln d)^{\nicefrac{3}{2}}\right)$,
\begin{align*}
\inf_A \sup_P \RT(A, P) &\geq
    \frac{1}{4} T^{\nicefrac{2}{3}} (d-1)^{\nicefrac{1}{3}} 
    \left(\ln\frac{1}{\deltaIR}\right)^{\nicefrac{1}{3}}
    - 2T^{\nicefrac{1}{2}} \in \Omega\left(T^{\nicefrac{2}{3}} d^{\nicefrac{1}{3}}\left(\ln\frac{1}{\deltaIR}\right)^{\nicefrac{1}{3}}\right).
\end{align*}
Next, suppose $\qmin \geq q^{0}_T \defeq 
T^{\nicefrac{-1}{3}} (d-1)^{\nicefrac{-2}{3}} \left(\ln\left(1/\deltaIR\right)\right)^{\nicefrac{1}{3}}$.
Then for $T \in \Omega\left(d\left(\ln (1/\eta)\right)^2(\ln d)^{\nicefrac{3}{2}}\right)$,
\begin{align*}
\inf_A \sup_{P;\; \qmin \geq q^{0}_T} \RT(A, P) &\geq
    \frac{1}{4} \sqrt{\frac{T}{q_{\min}}\ln\frac{1}{\deltaIR}} - 2T^{\nicefrac{1}{2}} 
     \in \Omega\left(\sqrt{\frac{T}{q_{\min}}\ln\frac{1}{\deltaIR}}\right).
\end{align*}

\end{theorem}

Comparing this with  Corollary~\ref{cor:regret}, we see that our algorithm is minimax optimal, up to constants and polylog terms.
This is the case even when
 $\qmin$ is larger than $\tilde\Omega(T^{\nicefrac{-1}{3}}d^{\nicefrac{-2}{3}})$
 where $\sqrt{T}$ rates are possible.
As we mentioned at the end of Section~\ref{sec:contributions},
our proof reveals interesting properties about the structure of an optimal policy; we discuss these in detail at the end of this section.

\begin{proof}[Proof of Theorem~\ref{thm:lb}]
Unlike typical proofs of lower bounds
in stochastic bandit settings, which usually rely on
hypothesis testing arguments, our result stems from the buyers' uncertainty about their values. 
To demonstrate this,
we will construct a representative problem instance and show that any algorithm will do poorly on
this instance.

\paragraph{Construction.}
For all types $j \in [d]$, we 
set the \emph{ex-post} value distribution to be
$\cD_j = {\rm Unif}(1-2/\sqrt{T}, 1)$.
Hence, for all $j\in[d]$, $\thetaj = 1-1/\sqrt{T}$.
Next, we define the type distribution $\cP$ as shown below.
Here $q < 1/d$ is a parameter we will specify later in the proof.
\begin{align*}
\forall\,j \in \{1,\dots, d-1\}, \, q_j = \Pr_{i \sim \cP}[i = j] = q,
\hspace{0.4in}
q_d = \Pr_{i\sim\cP}[i=d] = 1-q(d-1).
\numberthis
\label{eqn:construction}
\end{align*}
We will use the following threshold functions for each buyer of each type.
Recall that $\Phi_{i,t-1} = \{v; (i,v) \in \sigmatmo\}$ denotes the reviews in $\sigmatmo$ left by
customers of type $i$.
\begin{align*}
\tau_t(\sigmatmo, i) = \max\left\{\frac{1}{|\Phi_{i,t-1}|} \sum_{v\in \Phi_{i,t-1}} v
    \,-\, \sqrt{\frac{1}{2|\Phi_{i,t-1}|}\ln\frac{1}{\deltaIR} }, 0\right\}.
\end{align*}
Note that $\tau_t(\sigmatmo, i_t)$ is larger than ${\rm LB}_t$ as defined in
Definition~\ref{defn:deltapessimism} and satisfies the
$\eta-$pessimistic agents' assumption.
We will also assume that the seller knows the type distribution $\cP$; this additional
information can only help the seller.
Despite this, we show that if buyers choose conservative threshold functions, $T^{\nicefrac{2}{3}}$
regret is unavoidable.

The optimal price for the above construction is
$\pstar = \theta_1 = \dots = \theta_d = 1-1/\sqrt{T}$.
The seller could simply set this price if all buyers knew their \emph{ex-ante} values.
However, when buyers learn their values from past observations, the confidence of their estimates
shrinks only with the number of observations of their type.
In particular, if $q$ is very small, then a seller might find it beneficial to ignore customers of
the first $d-1$ types and set the highest possible price that can still attract customers of type
$d$. On the other hand,
if $q$ is large, the higher price may not warrant the revenue foregone by ignoring the first $d-1$
types. By carefully choosing $q$, we can balance these trade-offs to obtain the tightest
lower bound.
We have set the value of all types to be equal in this construction to simplify some of our
calculations, but it is not hard to see how this phenomenon affects pricing decisions for the seller.

\paragraph{Set up and notation.}
For brevity, we use the following notation for the sample mean of observations,
the number of observations, 
and the threshold function for type $i$ on round $t$.
\begin{align*}
&\EEitmo \defeq \frac{1}{|\Phi_{i,t-1}|} \sum_{v\in \Phi_{i,t-1}} \hspace{-0.10in}v,
\hspace{0.8in}
\Ntildeitmo \defeq  |\Phi_{i,t-1}|,
\\
&\tau_{i,t} \defeq \tau_t(\sigmatmo, i) = \max\left\{\EEitmo 
    - \sqrt{\frac{1}{2\Ntildeitmo}\ln\left(\frac{1}{\deltaIR} \right)}, 0\right\}.
\numberthis
\label{eqn:taushorthand}
\end{align*}
Next, let $\taupt$  denote the maximum of the threshold functions of the first $d-1$ types on
round $t$ and
$\jpt$ denote the corresponding maximizer.
\begin{align}
\jpt = \argmax_{j\in \{1,\dots,d-1\}} \tau_{j,t},
\hspace{0.4in}
\taupt = \tau_{\jpt, t}.
\label{eqn:jpt}
\end{align}
Recall that on each round $t$, a seller's policy chooses a price $\pt$ based on all past information
$\sigmatmo$ and possibly some source of external randomness.
We next define $\Wonet, \Wtwot, \Wthrt$ below based on how $\pt$ compares to the threshold
functions:
\begin{align}
\Wonet = \Ind\mathbbm(\pt \leq \taupt),
\hspace{0.2in}
\Wtwot = \Ind\mathbbm(\taupt < \pt \leq \tau_{d,t}),
\hspace{0.2in}
\Wthrt = \Ind\mathbbm(\forall\,j\in[d], \; \tau_{j,t} < \pt).
\label{eqn:Wit}
\end{align}
Here, $\Wonet$ is $1$ when the price $\pt$ is smaller than the thresholds for any of the first $d-1$
types,
$\Wtwot$ is $1$ when $\pt$ is larger than the thresholds for all $d-1$ types but smaller than
the threshold $\tau_{d,t}$ for type $d$ (note that $\Wtwot$ can be $1$ only when
$\taupt < \tau_{d,t}$), and
$\Wthrt$ is $1$ when $\pt$ is larger than all thresholds.
It is easy to verify that exactly one of $\Wonet, \Wtwot, \Wthrt$ is $1$ on any given round.

\paragraph{Lower bounding the instantaneous regret.}
We can decompose the expected revenue $\rev_t = \pt b_t$ on round $t$, conditioned on the price $\pt$ and history
$\sigmatmo$ as follows.
\begin{align*}
\E[\revt|\sigmatmo, \pt] &= 
    \E[\pt \bt|\sigmatmo, \pt]\\
    &=
    \sum_{i\in [d]} \pt \E[\bt|\sigmatmo, \pt,\,i_t=i]\Pr[i_t=i|\sigmatmo, \pt] \\
    &= \sum_{i\in [d]} \pt \mathbbm{1}(\pt\leq \tauiit) q_i
\\
&= \sum_{i\in [(d-1)]} \pt \mathbbm{1}(\pt\leq \tauiit) q
     + \pt \mathbbm{1}(\pt\leq \taudt) \qd.
\numberthis \label{eqn:revtdecomp}
\end{align*} 
In the third step
we have used the fact that that the probability of appearance of a type does not depend on the history or the price chosen,
hence
$\Pr[i_t=i|\sigmatmo, \pt] = \Pr_{j\sim\cP}[j=i] = q_i$.
Second, we note that for a customer of type $i$, they will purchase if and only if the price is smaller than their threshold; therefore $\E[\bt|\sigmatmo, \pt,\,i_t=i] = \mathbbm{1}(\pt\leq \tauiit)$.
The following lemma upper bounds $\E[\revt|\pt]$ in terms of the $\Wit$ terms defined
in~\eqref{eqn:Wit}.

\begin{lemma}\label{lem:revtub}
$\E[\revt|\sigmatmo, \pt] \leq \Wonet\taupt + \Wtwot\qd\taudt$.
\end{lemma}
\begin{proof}[Proof of Lemma~\ref{lem:revtub}]
We will consider four exhaustive cases for $\pt$ and analyze the right-hand side of the
inequality in the claim as a function of $\Wonet$ and $\Wtwot$, which we denote as ${\rm RHS}(\Wonet,\Wtwot)$.
\begin{enumerate}
\item \underline{$\pt \leq \min\left\{\tau_t', \tau_{d,t}\right\}$:} Here, $\Wonet=1$ and $\Wtwot=0$. Using~\eqref{eqn:revtdecomp}, we obtain
$\E[\revt|\pt] \leq (d-1)p_t q + p_tq_d = p_t \leq \tau_t' = {\rm RHS}(1,0)$.

\item \underline{$\taupt < \pt \leq \taudt$:}
Here, $\Wonet=0$ and $\Wtwot=1$. Using~\eqref{eqn:revtdecomp}, 
$\E[\revt|\pt] = p_tq_d \leq \taudt q_d = {\rm RHS}(0,1)$.

\item \underline{$\taudt < \pt \leq \taupt$}:
Here, $\Wonet=1$ and $\Wtwot=0$.
Using~\eqref{eqn:revtdecomp}, we obtain 
\[\E[\revt|\pt] \leq \pt(d-1)q < \pt \leq \taupt = {\rm RHS}(1,0).\]
Some of terms in the first summation in~\eqref{eqn:revtdecomp} may be $0$, but we can bound it
by $\pt(d-1)q$ regardless.

\item \underline{$p_t > \max\left\{\tau_t', \tau_{d,t}\right\}$:}
Here, $\Wonet=\Wtwot=0$.
Using~\eqref{eqn:revtdecomp}, 
$\E[\revt|\pt] = 0 = {\rm RHS}(0,0)$. 
\end{enumerate}\end{proof}
Equipped with this lemma, we can now lower bound the instantaneous regret on round $t$ conditioned on the price $\pt$ and history $\sigmatmo$, which we denote as $\E[r_t|\sigmatmo, \pt]$: 
\begin{align}
\E[r_t|\sigmatmo, \pt] &=  \E[\pstar - \revt|\sigmatmo, \pt] \nonumber\\
&\geq \Wonet \cdot (\pstar - \taupt) + \Wtwot \cdot \left(\pstar(d-1)q + \qd(\pstar-\taudt)\right)
        + \Wthrt \cdot \pstar.\label{eq:lower_bnd_w}
        \end{align}

Recall that $\jpt$ is the index such that $\taupt = \tau_{\jpt, t}$ as defined in~\eqref{eqn:jpt} and $\Ntildejptmo$ is the number of observations of type $\jpt$ in $\sigma_{t-1}$ as defined
in~\eqref{eqn:taushorthand}.
We can further lower bound Equation~\eqref{eq:lower_bnd_w} by using the fact that \begin{equation}\pstar - \taut' = \pstar - \taujptt =
\pstar - \max\left\{\EEjpttmo - \sqrt{\frac{1}{2\Ntildejptmo}\ln\left(\frac{1}{\deltaIR}\right)}, 0 \right\}.\label{eq:taut_bnd}\end{equation}
Since the support of each \emph{ex-post} value distribution is bounded within an
$\pm 1/\sqrt{T}$ interval of $\pstar$, Equation~\eqref{eq:taut_bnd} implies that \begin{align}\pstar - \taut' &\geq \pstar - \max\left\{p^* + \frac{1}{\sqrt{T}} - \sqrt{\frac{1}{2\Ntildejptmo}\ln\left(\frac{1}{\deltaIR}\right)}, 0 \right\}\nonumber\\
&= \min\left\{\sqrt{\frac{1}{2\Ntildejptmo}\ln\left(\frac{1}{\deltaIR}\right)} - \frac{1}{\sqrt{T}}, p^* \right\}\nonumber\\
&= \min\left\{\sqrt{\frac{1}{2\Ntildejptmo}\ln\left(\frac{1}{\deltaIR}\right)}, 1\right\}-\frac{1}{\sqrt{T}}. \label{eq:min}\end{align} 
The same argument guarantees that $
\pstar - \taudt 
\geq - \frac{1}{\sqrt{T}}.$ Combining this inequality with Equations~\eqref{eq:lower_bnd_w} and \eqref{eq:taut_bnd}, and recalling that $\Wonet + \Wtwot + \Wthrt = 1$, we have that \begin{equation}\E[r_t|\sigmatmo, \pt]\geq \Wonet\min\left\{\sqrt{\frac{1}{2\Ntildejptmo}\ln\left(\frac{1}{\deltaIR}\right)}, 1\right\}
    + \Wtwot\pstar(d-1)q  + \Wthrt\pstar - \frac{1}{\sqrt{T}}.\label{eqn:rtupperbound}\end{equation}

\paragraph{Upper bounding $\Ntildejptmo$.}
To convert the above instantaneous bound to a lower bound on the cumulative regret, we will need to
control $\Ntildejptmo$ which counts the number of reviews in $\sigmatmo$ by customers of type
$\jpt$. 
Observing that $\jpt\in[(d-1)]$ which means that the appearance probability of $\jpt$
is $q$, we define the following event $\cE$ below.
Lemma~\ref{lem:cEtbound} upper bounds the probability of this event.
\begin{align}
\label{eqn:lbevent}
\cE = \left\{
\forall\, j\in [(d-1)], \forall \, t \leq T, \;
\Ntildejtmo \leq 2q(T-1)\
\right\}.
\end{align}

\begin{lemma}\label{lem:cEtbound}
Let $T\geq \frac{3}{q}\ln(2d) + 1$. Then,
$\PP[\cE] \geq 1/2$.
\end{lemma}
\begin{proof}[Proof of Lemma~\ref{lem:cEtbound}]
Note that \[\Ntildeitmo = \sum_{s=1}^{t-1} \indfone(b_s=1, i_s=i)\] counts the number of times
a customer of type $i$ made a purchase.
Let \[\Nitmo=\sum_{s=1}^{t-1} \indfone(i_s=i)\] be the number of times
a customer of type $i$ arrived.
Since $N_{i,T-1}\geq\Ntildeitmo$, the Chernoff bound implies that
\[
\Pr[\exists t \leq T \text{ such that } \Ntildeitmo > 2q(T-1)]
\leq
\Pr[N_{i,T-1} > 2q(T-1)]
\leq \exp\left(\frac{q(T-1)}{3}  \right)
\leq \frac{1}{2d}.
\]
The last step uses the condition on $T$. The claim follows via a union bound over $j\in
[d-1]$.
\end{proof}

\paragraph{Lower bound on cumulative regret.}
We are now ready to lower bound regret.
By Equation~\eqref{eqn:rtupperbound},
\[\E[R_T] \geq 
    \E\left[\sum_{t = 1}^T \Bigg(
         \underbrace{\Wonet\min\left\{\sqrt{\frac{1}{2\Ntildejptmo}\ln\left(\frac{1}{\deltaIR}\right)}, 1\right\}
        + \Wtwot\pstar(d-1)q  + \Wthrt\pstar}_{\overline{r}_t} - \frac{1}{\sqrt{T}} \Bigg) \right].\]

Conditioning on the event $\cE$,  
\[\E[R_T] \geq 
    - \sqrt{T} + \E\left[\left. \sum_{t = 1}^T \overline{r}_t \, \right| \, \cE\right]\Pr[\cE] 
    + \E\left[\left. \sum_{t = 1}^T \overline{r}_t \, \right| \, \cE^c\right]\Pr[\cE^c] \] 
    and by Lemma~\ref{lem:cEtbound},
    \begin{align*}\E[R_T] &\geq - \sqrt{T}
       +
    \frac{1}{2}\E\left[\left. \sum_{t = 1}^T \overline{r}_t \, \right| \, \cE\right]
    - T\cdot \frac{1}{\sqrt{T}}
\\
&\geq  -2\sqrt{T}
       +
    \frac{1}{2} \E\left[\sum_{t = 1}^T\left. 
\Wonet \min\left\{\sqrt{\frac{1}{4qT}\ln\left(\frac{1}{\deltaIR}\right)}, 1\right\}
        + \Wtwot\pstar(d-1)q  + \Wthrt\pstar \, \right| \, \cE
 \right].\end{align*}
 For $T > \frac{1}{4q}\ln\left(\frac{1}{\deltaIR}\right)$,
 \[\E[R_T] \geq -2\sqrt{T}
       +
    \frac{1}{2}\E\left[\left. \sum_{t = 1}^T 
\Wonet \sqrt{\frac{1}{4qT}\ln\left(\frac{1}{\deltaIR}\right)}
        + \Wtwot\pstar(d-1)q  + \Wthrt\pstar
\, \right| \, \cE \right].\]
We will use the notation $\Mit = \sum_{s=1}^t \Wis$ for $i\in\{1,2,3\}$ which counts the number of times each
$\Wis$ was $1$ in the first $t$ rounds. Note that
$\Monet+\Mtwot + \Mthrt = t$ since exactly one of
$\Wones,\Wtwos,\Wthrs$ is $1$ on any round $s$. With this notation, we have that
\begin{align}
    \E[R_T] \geq - 2\sqrt{T}
       +
    \frac{1}{2}\E\left[\left.
        \MoneT\sqrt{\frac{1}{4qT}\ln\left(\frac{1}{\deltaIR}\right)}
        + \MtwoT\pstar(d-1)q  + \MthrT\pstar
 \, \right| \, \cE
 \right].
 \label{eqn:RTMiT}
\end{align}

We note that $\MoneT,\MtwoT,\MthrT$ are random quantities that depend on the execution of
the algorithm. However, we can use the fact that they are non-negative and that
 $\MoneT+\MtwoT+\MthrT = T$
to obtain a lower bound as follows.
\[\E[R_T] \geq - 2\sqrt{T} +
    \frac{1}{2}
    \inf_{\substack{x_1, x_2, x_3 > 0 \\ x_1+x_2+x3=T}} \left(
        x_1\sqrt{\frac{1}{4qT}\ln\left(\frac{1}{\deltaIR}\right)}
        + x_2\pstar(d-1)q  + x_3\pstar\right).\]
As $(d-1)q\leq 1$, for any choice
$(x'_1, x'_2, x'_3)$ for $(x_1, x_2, x_3)$  such that $x'_3>0$, we can obtain a lower value for the
term in parentheses via $(x'_1, x'_2 + x'_3, 0)$. Therefore, the above expression simplifies to:
\begin{align}\E[R_T] \geq - 2\sqrt{T} +
    \frac{1}{2}
    \inf_{0\leq x \leq T} \left(
        x\sqrt{\frac{1}{4qT}\ln\left(\frac{1}{\deltaIR}\right)}
        + (T-x)\pstar(d-1)q \right).
\label{eqn:RTinflinear}
\end{align}
Finally, we are taking the infimum of a linear function in the bounded interval
$[0, T]$, so the infimum lies at one of the end points $x=0$ or $x=T$. Therefore, \begin{align}\E[R_T] &\geq - 2\sqrt{T} +
    \frac{1}{2} \min\left\{ 
        \sqrt{\frac{T}{4q}\ln\left(\frac{1}{\deltaIR}\right)}\;,\;
        \;
        T\pstar(d-1)q
    \right\}\nonumber\\
    &\geq - 2\sqrt{T} +
    \frac{1}{2} \min\left\{ 
        \sqrt{\frac{T}{4q}\ln\left(\frac{1}{\deltaIR}\right)}\;,\;
        \;
        T\pstar(d-1)q
    \right\}.\label{eqn:RTminbound}\end{align}

\paragraph{Putting it all together.}
To complete the proof, first note that for all $T\geq 4$, $\pstar \geq 1/2$;
hence, the second term inside the $\min$ can be upper bounded by $\frac{1}{2}T(d-1)q$.
To obtain a $\qmin$ independent bound, 
we set $q= T^{\nicefrac{-1}{3}} (d-1)^{\nicefrac{-2}{3}} 
    (\ln(1/\deltaIR))^{\nicefrac{1}{3}}$ to obtain the first result of the theorem.

Next, since $q_{\min} = q$ for this problem, we have that when \[\qmin>q^0_T = T^{\nicefrac{-1}{3}} (d-1)^{\nicefrac{-2}{3}} \left(\ln\left(1/\deltaIR\right)\right)^{\nicefrac{1}{3}},\] the minimum is the
first of the two terms in~\eqref{eqn:RTminbound}.
This leads to our second lower bound.
\end{proof}

Our construction uses $\pstar$ close to $1$ to simplify some of the calculations in the
analysis, but
a similar analysis is possible for any $\pstar$ bounded away from $0$.
Second, while our construction sets the ex-ante value $\thetaj$ to be the same for all
types, a similar result can be shown in cases where a low probability type has ex-ante value similar
to or larger than the ex-ante value of high probability types.
Third, recall that we have assumed in this proof that the seller knows the type distribution $\cP$.
If it is unknown,
as was shown in our upper-bound analysis, the seller only really needs to estimate the low
probability types and the expected revenue when targeting the remaining types,
both of which can be
done at rates $T^{\nicefrac{1}{3}}$ and  $T^{\nicefrac{1}{2}}$ respectively without having to learn
$\cP$ entirely. The $T^{\nicefrac{2}{3}}$ bottleneck arises as the seller needs to wait for the
buyers' estimates of their values become accurate.

We also make the following observation via Equations~\eqref{eqn:RTMiT}--\eqref{eqn:RTminbound}.
Intuitively, $\MoneT$ in~\eqref{eqn:RTMiT} denotes the number of times the seller's policy targeted the low probability types, $\MtwoT$ denotes the number of times it targeted the high probability type while ignoring the low probability types, and $\MthrT$ is the number of times it targeted none of the types.
Equation~\eqref{eqn:RTinflinear} states that any reasonable policy will never ignore all customer types, choosing $\MthrT=0$.
On the other hand, the fact that the infimum in~\eqref{eqn:RTminbound} lies in one of two extremes $(M_{1,T}, M_{2,T}) \in \{(0,T), (T,0)\}$ indicates that any reasonable policy cannot do significantly better than a policy which chooses ahead of time to target all customer types or only focus on the high probability types.
Intuitively, this means that the seller's policy can decide ahead of time which customers it wants to ignore due to a low probability of appearance. 
In other words, it does not
significantly help to change which types you target on different rounds based on their appearance probability.
Interestingly, this is precisely the behavior of our algorithm as well; it uses a small initial phase of at most $T^{\nicefrac{1}{3}}$ rounds to identify and eliminate low probability types.
From thereon, it only targets the remaining high probability types.

\section{Conclusion}\label{sec:conclu}

We proposed no-regret online pricing strategies when both sides of the market learn from reviews.
Our algorithm strategically sets lower prices during its early phase to boost sales from customers with rare types and high values. Reviews from the early phase benefit future buyers in the long run. Our algorithm carefully trades off the revenue loss due to discounts from the initial phase and future gains.
Our lower bound demonstrates that our algorithm is optimal up to lower order and constant terms.
To the best of our knowledge, this is the first result on online pricing when both the seller learns to price and buyers with different types learn from reviews. 

\paragraph{Future directions.} Many questions remain open for future research. We assumed that purchases always come with a noisy review. An interesting direction would be providing pricing strategies when the reviews are left with varying probabilities, which mimics real-world buyer behaviors.

We studied myopic buyers who make their purchase decisions based on estimates of their \emph{ex-ante} values from historical reviews, regardless of the seller's policy. What if the buyers appear over several rounds  and may behave strategically to purchase at lower future prices?

We take a frequentist perspective on this problem. It is also possible to take a Bayesian view of this problem and impose a prior on the \emph{ex-ante} value 
so that the buyer starts with some prior information. We expect adapting our main proof intuitions to that setting is possible. The main differences would be: \emph{(i)} we would use Bayesian credible intervals instead of frequentist confidence intervals for the  $\eta$-pessimism definition, \emph{(ii)} we would control the Bayes' risk when estimating the \emph{ex-ante} values instead of frequentist concentration arguments, and \emph{(iii)} our final regret could have a nuanced dependence on this prior which may offer tighter bounds.

Another direction would be to explore the case where the seller does not know the buyers' \emph{ex-ante} values. The key challenge would be related to the regret benchmark: we compete with the optimal price if the buyers knew their own \emph{ex-ante} values and bought whenever their \emph{ex-ante} value was above the price (thus, the buyers are not learning). To compete with this benchmark, we require unbiased estimates of the revenue of different prices if the buyers bought when their \emph{ex-ante} value was above the price. Computing these unbiased estimates is challenging: if a buyer does not buy on a given round, the algorithm does not learn their type, so it cannot tell whether the buyer has a low \emph{ex-ante} value or he has a high value but a low confidence bound. If the seller knows the buyers' \emph{ex-ante} values, we can circumvent this subtle challenge, as we explain in the proof sketch of Theorem~\ref{thm:regret}. However, this is not possible if the \emph{ex-ante} values are unknown.

\bibliographystyle{plainnat}
\bibliography{ref}

\appendix
\section{Additional details about $\eta$-pessimistic agents}\label{app:LCB-def-proof}

Intuitively, in Definition~\ref{defn:deltapessimism}, $\LB_t$ serves as a lower confidence bound on the buyer's value who arrives at round $t$. The buyers can be confident that, \emph{regardless} of the policy used by the seller, with probability $1-\eta$, for all rounds $t \in [T]$, $\theta_{i_t} \geq \LB_t$. We show this formally below.

\begin{lemma}\label{thm:eta-LCB}
Denote the type of the buyer who arrives at round $t$ as $i_t$.
On all rounds $t$, with probability at least $1-\eta$, $\LB_t \leq \theta_{i_t}$.
\end{lemma}

\begin{proof}
Let us consider a sequence of $T$ rewards $\{\tilde v_{i1}, \cdots \tilde v_{iT}\}$ for each buyer type $i \in [d]$ generated beforehand, where each reward is a random reward sample drawn from $\cD_i$. 
Each time a buyer with type $i$ arrives and makes a purchase, it obtains an ex-post value from the reward sequence $\{\tilde v_{i1}, \cdots \tilde v_{iT}\}$ in order. For example, if the type of the buyer who arrives on round $t$ is $i_t$, then if that buyer makes a purchase, their ex-post value will be $\tilde v_{i_t, |\Phi_{i, t}|+1}$.

First, we will show that $\Pr\left(\LB_t > \theta_{j} \mid i_t = j\right) \leq \eta $ for any $j \in [d]$. At any round $t$, notice that if $|\Phi_t| = 0$, then $\LB_t = 0$, the conclusion trivially holds since $\theta_j>0$ for all $j \in [d]$. When $|\Phi_t| > 0$:
\begingroup
\allowdisplaybreaks
\begin{align*}
    & \Pr\left(\LB_t > \theta_{j} \bigg| i_t = j\right) \\
    &= \Pr\left(\max\left\{\,0,\;\frac{1}{|\Phi_{t}|}\sum_{v \in \Phi_t}v - \sqrt{\frac{1}{2|\Phi_{t}|}\ln \frac{t}{\eta}}\right\}  > \theta_{j} \bigg| i_t = j\right) \\
    &= \Pr\left(\frac{1}{|\Phi_{t}|}\sum_{v \in \Phi_t}v - \sqrt{\frac{1}{2|\Phi_{t}|}\ln \frac{t}{\eta}}  > \theta_{j} \bigg| i_t = j\right) \\
    &= \Pr\left(\frac{1}{|\Phi_{t}|}\sum_{s=1}^{|\Phi_t|}\tilde v_{js} - \sqrt{\frac{1}{2|\Phi_{t}|}\ln \frac{t}{\eta}}  > \theta_{j} \bigg| i_t = j\right) \\
    &\leq \Pr\left(\exists \ell \in [t-1], \text{s.t. } \frac{1}{\ell}\sum_{s=1}^{\ell} \tilde v_{js} - \sqrt{\frac{1}{2\ell}\ln \frac{t}{\eta}}> \theta_{j} \bigg| i_t = j \right) \\
    &\leq \sum_{\ell=1}^{t-1}\Pr\left(\frac{1}{\ell} \sum_{s=1}^\ell \tilde v_{js} - \sqrt{\frac{1}{2\ell}\ln \frac{t}{\eta}}> \theta_{j } \bigg| i_t = j\right).
\end{align*}
\endgroup

Here, the second step uses the fact that $\theta_j\geq 0$.
In the fifth step, we have used the fact that $|\Phi_t|$ is a random quantity, which depends on the specific algorithm, but with support $[(t-1)]$.
The last step follows from a union bound over $(t-1)$ rounds. 

Note that for any fixed $j\in [d]$, the event $\frac{1}{\ell} \sum_{s=1}^\ell \tilde v_{js} - \sqrt{\frac{1}{2\ell}\ln \frac{T}{\eta}}> \theta_{j }$ is independent of the value of $i_t$. Therefore, by Hoeffding inequality, for any $\ell \in [t-1]$ and $j \in [d]$, we have that 
\begin{align*}
    \Pr\left(\frac{1}{\ell} \sum_{s=1}^\ell \tilde v_{js} - \sqrt{\frac{1}{2\ell}\ln \frac{t}{\eta}}> \theta_{j} \bigg| i_t = j\right) =   \Pr\left(\frac{1}{\ell} \sum_{s=1}^\ell \tilde v_{js} - \sqrt{\frac{1}{2\ell}\ln \frac{t}{\eta}}> \theta_{j} \right) \leq \frac{\eta}{t}.
\end{align*}

Putting this together we have:
\begin{align*}
  \Pr\left(\LB_t > \theta_{j} \mid i_t = j\right) \leq (t-1) \frac{\eta}{t} \leq \eta. 
\end{align*}

Lastly, by the law of total probability,
\begin{align*}
    &\Pr\left(\LB_t > \theta_{i_t} \right) = \sum_{j \in [d]} \Pr\left(\LB_t > \theta_{j} \mid i_t = j\right) \cdot \Pr\left(i_t = j\right)  \leq \sum_{j \in [d]} \eta \cdot \Pr\left(i_t = j\right) \leq \eta,
\end{align*}
which completes the proof.
\end{proof}
\section{Additional proofs about regret upper bound (Section~\ref{sec:upper_proof})}
\label{app:upper_proof}
\label{app:lemmas}

\Ztwo*

\begin{proof}
    First, we bound $Z_2$ as follows: \[Z_2 = \sum_{t > t_{\lambda}}p^*([d])\Ind\left(\theta_{i_t} \geq p^*([d]) \text{ and } i_t \not\in Q\right) \leq \sum_{t > t_{\lambda}}p^*([d])\Ind\left(i_t \not\in Q\right) \leq \sum_{t > t_{\lambda}}\Ind\left(i_t \not\in Q\right).\]
    Recall from Algorithm~\ref{alg:type_elim} that $\overline{q}_i$ is the fraction of times that type $i$ appears in phase 1 and let $\cG$ be the event that for all $i \in [d]$ such that $q_i \geq \lambda$, we have that $\overline{q}_i \geq \frac{3\lambda}{4}$, which means that $i \in Q$. In other words, when $\cG$ happens, $[d] \setminus Q \subseteq \{q_i : q_i < \lambda\}$. In Lemma~\ref{lem:G}, we prove that $\Pr[\cG^c] \leq \frac{1}{T},$ so $\E[Z_2] \leq \E[Z_2 \mid \cG] + T\Pr[\cG^c] \leq \E[Z_2 \mid \cG] + 1$.

    Next, since $Q$ is a random variable, we condition on it as well: \[\E[Z_2 \mid \cG] = \sum_{Q' \subseteq [d]} \E[Z_2 \mid Q = Q',  \cG]\Pr[Q = Q' \mid \cG].\] If $[d]\setminus Q' \not\subseteq \left\{q_i : q_i < \lambda\right\}$, then $\Pr[Q = Q' \mid \cG] = 0$. Moreover, for any $Q'$ such that $[d]\setminus Q' \subseteq \left\{q_i : q_i < \lambda\right\}$, 
    \begin{align*}\E[Z_2 \mid Q = Q', \cG] &\leq \E\left[\left.\sum_{t > t_{\lambda}}\Ind\left(i_t \not\in Q'\right) \, \right| \, Q = Q', \cG\right]\\
            &= \sum_{t > t_{\lambda}}\Pr\left[i_t \not\in Q' \mid Q = Q',  \cG\right]\\
            &=\sum_{t > t_{\lambda}}\sum_{i \not\in Q'}\Pr\left[i_t = i \mid Q = Q',\cG\right].
            \end{align*}
    The event $(Q = Q' \land \cG)$ depends only on the first $t_{\lambda}$ timesteps, so it is independent of the event that $i_t = i$ for $t > t_{\lambda}$. Therefore, \[\E[Z_2 \mid Q = Q',  \cG] = \sum_{t > t_{\lambda}}\sum_{i \not\in Q'}\Pr\left[i_t = i\right].\]

If $q_{\min} > 2\lambda$, then $\left\{q_i : q_i < \lambda\right\} = \emptyset$, so the only $Q'$ such that $[d]\setminus Q' \subseteq \{q_i : q_i < \lambda\}$ is $Q' = [d]$. In this case,  \[\sum_{t > t_{\lambda}}\sum_{i \not\in Q'}\Pr\left[i_t = i\right] =0,\] so $\E[Z_2 \mid \cG] =0$ and finally, $\E[Z_2] \leq 1.$
    
Otherwise, $q_{\min} \leq 2\lambda$, so \[\sum_{t > t_{\lambda}}\sum_{i \not\in Q'}\Pr\left[i_t = i\right] \leq \sum_{t > t_{\lambda}}\sum_{i \not\in Q'}\lambda \leq Td\lambda,\] so \[\E[Z_2 \mid \cG] \leq Td\lambda \sum_{Q' \subseteq [d]} \Pr[Q = Q' \mid \cG] \leq Td\lambda\] and finally, $\E[Z_2] \leq Td\lambda + 1.$
        \end{proof}

\Zthree*

\begin{proof}
In this proof we bound \begin{equation}Z_3 = \sum_{t > t_{\lambda}}p^*([d])\Ind\left(\theta_{i_t} \geq p^*([d]) \text{ and } i_t \in Q\right) - \sum_{t > t_{\lambda}}p^*(Q)\Ind\left(\theta_{i_t} \geq p^*(Q) \text{ and } i_t \in Q\right).\label{appeq:Z3}\end{equation}
We begin by conditioning the first term of Equation~\eqref{appeq:Z3} on $Q$ since it is a random variable: \begin{align}&\E\left[\sum_{t > t_{\lambda}}p^*([d])\Ind\left(\theta_{i_t} \geq p^*([d]) \text{ and } i_t \in Q\right)\right]\nonumber\\
            =&\sum_{Q' \subseteq [d]}p^*([d])\E\left[\sum_{t > t_{\lambda}}\Ind\left(\theta_{i_t} \geq p^*([d]) \text{ and } i_t \in Q'\right) \mid Q = Q'\right]\Pr[Q = Q']\nonumber\\
            =&\sum_{Q' \subseteq [d]}\sum_{t > t_{\lambda}}p^*([d])\Pr\left[\theta_{i_t} \geq p^*([d]) \text{ and } i_t \in Q' \mid Q = Q'\right]\Pr[Q = Q'].\label{appeq:Z3_Q}\end{align} The event that $Q = Q'$ only depends on the first $t_{\lambda}$ timesteps, so it is independent of the event $\left(\theta_{i_t} \geq p^*([d]) \land i_t \in Q'\right)$ for $t > t_{\lambda}$. Therefore, for $t > t_{\lambda}$, \begin{align*}p^*([d])\Pr\left[\theta_{i_t} \geq p^*([d]) \text{ and } i_t \in Q' \mid Q = Q'\right] &= p^*([d]) \Pr\left[\theta_{i_t} \geq p^*([d]) \text{ and } i_t \in Q' \right]\\
        &\leq \max_{p \in [0,1]}p\Pr\left[\theta_{i_t} \geq p \text{ and } i_t \in Q' \right]\\
        &= p^*(Q')\Pr\left[\theta_{i_t} \geq p^*(Q') \text{ and } i_t \in Q' \right].\end{align*} Combining this fact with Equation~\eqref{appeq:Z3_Q}, we have that \begin{align}&\E\left[\sum_{t > t_{\lambda}}p^*([d])\Ind\left(\theta_{i_t} \geq p^*([d]) \text{ and } i_t \in Q\right)\right]\nonumber\\
        \leq \, &\sum_{Q' \subseteq [d]}\left(T-t_{\lambda}\right)p^*(Q')\Pr_{i \sim \cP}\left[\theta_i \geq p^*(Q') \text{ and } i \in Q' \right]\Pr[Q = Q'].\label{appeq:first_term}\end{align}

Next, for the second term of Equation~\eqref{appeq:Z3}, 
        \begin{align*}&\E\left[\sum_{t > t_{\lambda}}p^*(Q)\Ind\left(\theta_{i_t} \geq p^*(Q) \text{ and } i_t \in Q\right)\right]\\
            =&\sum_{Q' \subseteq [d]}\E\left[\sum_{t > t_{\lambda}}p^*(Q')\Ind\left(\theta_{i_t} \geq p^*(Q') \text{ and } i_t \in Q'\right) \mid Q = Q'\right]\Pr[Q = Q'].\end{align*} As before, the event that $Q = Q'$ only depends on the first $t_{\lambda}$ timesteps, so it is independent of the event $\left(\theta_{i_t} \geq p^*(Q') \land i_t \in Q'\right)$ for $t > t_{\lambda}$. Therefore, \begin{align*}&\E\left[\sum_{t > t_{\lambda}}p^*(Q)\Ind\left(\theta_{i_t} \geq p^*(Q) \text{ and } i_t \in Q\right)\right]\\
            = \, &\sum_{Q' \subseteq [d]}\left(T - t_{\lambda}\right)p^*(Q')\Pr_{i \sim \cP}\left[\theta_i \geq p^*(Q') \text{ and } i \in Q'\right]\Pr[Q = Q'].\end{align*}
Combined with Equation~\eqref{appeq:first_term}, we have that $\E[Z_3] \leq 0$.
\end{proof}

\Zfour*

\begin{proof}
In this claim, we bound \begin{equation}Z_4 = \sum_{t > t_{\lambda}}p^*(Q)\Ind\left(\theta_{i_t} \geq p^*(Q) \text{ and } i_t \in Q\right) - \sum_{t > t_{\lambda}}p_{t}'\Ind\left(\theta_{i_t} \geq p_{t}' \text{ and } i_t \in Q\right)\label{appeq:Z4},\end{equation} where $p_t' = \min_{i\in S_t}\theta_{i}$. Beginning with the first term of this equation, for any $t > t_{\lambda}$, \begin{align*}&\E\left[
p^*(Q)\Ind\left(\theta_{i_t} \geq p^*(Q) \text{ and } i_t \in Q\right)\right]\\
= &\sum_{Q' \subseteq [d]}\E\left[
p^*(Q')\Ind\left(\theta_{i_t} \geq p^*(Q') \text{ and } i_t \in Q'\right) \mid Q = Q'\right]\Pr[Q = Q'].\end{align*} The event $\left(\theta_{i_t} \geq p^*(Q') \land i_t \in Q'\right)$ is independent of the event that $Q = Q'$, so \begin{align*}\E\left[
p^*(Q')\Ind\left(\theta_{i_t} \geq p^*(Q') \text{ and } i_t \in Q'\right) \mid Q = Q'\right] &= \E\left[
p^*(Q')\Ind\left(\theta_{i_t} \geq p^*(Q') \text{ and } i_t \in Q'\right) \right]\\
&= \rev(p^*(Q'), Q').\end{align*}
Therefore, \begin{align}\sum_{t > t_{\lambda}}\E\left[
p^*(Q)\Ind\left(\theta_{i_t} \geq p^*(Q) \text{ and } i_t \in Q\right)\right] &= \sum_{t > t_{\lambda}}\sum_{Q' \subseteq [d]} \rev(p^*(Q'), Q')\Pr[Q = Q']\nonumber\\
&= \sum_{t > t_{\lambda}}\E\left[\rev(p^*(Q), Q)\right]\label{appeq:Z4_first}.\end{align}

Moving on to the second term of Equation~\eqref{appeq:Z4}, we have that for any $t > t_{\lambda}$, \begin{align}&\E\left[
p_{t}'\Ind\left(\theta_{i_t} \geq p_{t}' \text{ and } i_t \in Q\right)\right]\nonumber\\
=\, &\sum_{Q' \subseteq [d]}\sum_{S' \subseteq Q'}\E\left[\left.
\min_{i \in S'} \theta_{i} \cdot \Ind\left(\theta_{i_t} \geq \min_{i \in S'} \theta_{i} \text{ and } i_t \in Q'\right) \, \right| \, Q = Q', S_t = S'\right]\Pr[Q = Q', S_t = S'].\label{appeq:conditionSQ}
\end{align}
The event that $Q = Q'$ only depends on the first $t_{\lambda}$ timesteps and the event that $S_t = S'$ only depends on the first $t - 1$ timesteps. Therefore, the event $\left(\theta_{i_t} \geq \min_{i \in S'} \theta_{i} \text{ and } i_t \in Q'\right)$ is independent of the event $\left(Q = Q' \text{ and } S_t = S'\right)$. This means that \begin{align*}&\E\left[\left .
\min_{i \in S'} \theta_{i} \cdot \Ind\left(\theta_{i_t} \geq \min_{i \in S'} \theta_{i} \text{ and } i_t \in Q'\right) \, \right| \, Q = Q', S_t = S'\right]\\
= \, &\E\left[
\min_{i \in S'} \theta_{i} \cdot \Ind\left(\theta_{i_t} \geq \min_{i \in S'} \theta_{i} \text{ and } i_t \in Q'\right)\right]\\
=\, &\rev\left(
\min_{i \in S'}\theta_{i}, Q'\right).\end{align*} Combined with Equation~\eqref{appeq:conditionSQ}, we have that 
\begin{align}\E\left[
p_{t}'\Ind\left(\theta_{i_t} \geq p_{t}' \text{ and } i_t \in Q\right)\right]&= \sum_{Q' \subseteq [d]}\sum_{S' \subseteq Q'}\rev\left(
\min_{i \in S'}\theta_{i}, Q'\right)\Pr[Q = Q', S_t = S']\nonumber\\
&= \E\left[\rev\left(
p_t', Q\right)\right].\label{appeq:Z4_second}\end{align}
Combining Equations~\eqref{appeq:Z4_first} and \eqref{appeq:Z4_second}, we have that \begin{equation}\E[Z_4] \leq \sum_{t > t_{\lambda}}\E\left[\rev\left(
p^*(Q), Q\right) - \rev\left(
p_t', Q\right)\right].\label{appeq:Z4_bound}\end{equation}

Next, for all $t > t_{\lambda}$, let $\cB_t$ be the event that:
    \begin{enumerate}
    \item $i_Q \in S_{t}$ and
    \item $\rev(p^*(Q), Q) - \rev\left(p_t', Q\right) \leq 4\rho_{t-1}$ (where $\rho_{t_{\lambda}} = 1$).
    \end{enumerate}
    Also, let $\cC_t = \bigcap_{s = t_{\lambda} + 1}^t \cB_s.$ In Lemma~\ref{lemma:SE}, we prove that $\Pr\left[\cC_t^c\right] \leq \frac{1}{T}.$ By Equation~\eqref{appeq:Z4_bound}, we have that \begin{align*}\E[Z_4] &\leq \E\left[\left. \sum_{t > t_{\lambda}}\rev\left(
p^*(Q), Q\right) - \rev\left(
p_t', Q\right) \, \right| \, \cC_T\right] + T\Pr\left[\cC_T^c\right]\\
&\leq \sum_{t > t_{\lambda}} 4\rho_{t-1} + 1\\
&\leq 5 + 4\sum_{t = 1}^T \sqrt{\frac{\ln(dT^2)}{2t}}\\
&\leq 5 + 4\sqrt{2T\ln(dT^2)}.\end{align*}
    \end{proof}

\Zfive*
    \begin{proof}
    On each round $t > t_{\lambda}$,
        recall that \[\text{LB}_{it}=\begin{cases} 0 &\text{if } \Phi_{it} = \emptyset\\
\max\left\{\frac{1}{|\Phi_{it}|}\sum_{v \in \Phi_{it}}v - \sqrt{\frac{1}{2|\Phi_{it}|}\ln \frac{T}{\eta}}, 0\right\} &\text{else.}\end{cases}\] Let $j_t = \argmin_{j \in S_t} \theta_{j}$, so $p_t' = \theta_{j_t}$. Then
        \begin{align*}Z_5 &= \sum_{t > t_{\lambda}}p_t'\Ind\left(\theta_{i_t} \geq p_t' \text{ and } i_t \in Q\right)- \sum_{t>t_{\lambda}} p_tb_t\\
        &= \sum_{t > t_{\lambda}}p_t'\Ind\left(\theta_{i_t} \geq p_t' \text{ and } i_t \in Q\right)- \sum_{t>t_{\lambda}} (p_t' + (p_t' - p_t))b_t\\
        &= \sum_{t > t_{\lambda}}p_t'\left(\Ind\left(\theta_{i_t} \geq p_t'  \text{ and } i_t \in Q\right) - b_t\right)+ \sum_{t > t_{\lambda}}(p_t' - p_t)b_t\\
        &\leq \sum_{t > t_{\lambda}}p_t'\Ind\left(\theta_{i_t} \geq p_t' \land i_t \in Q\land b_t = 0\right)+ \sum_{t > t_{\lambda}}(p_t' - p_t)b_t.\end{align*} Since $p_t \leq p_t'$, we have that \[Z_5 \leq \sum_{t > t_{\lambda}}p_t'\Ind\left(\theta_{i_t} \geq p_t' \land i_t \in Q\land b_t = 0\right)+ \sum_{t > t_{\lambda}}\left(p_t' - p_t\right).\] By definition of the pricing rule, if $i_t \in S_t$, then $b_t = 1$. Therefore, if $b_t = 0$, then either $i_t \not\in Q$ or $i_t \in Q \setminus S_t$. Since $S_t$ contains every $i \in Q$ with $i > j_t$, we can conclude that if $i_t \in Q \setminus S_t$, then $\theta_{i_t} < \theta_{j_t} = p_t'$. Therefore, $\Ind\left(\theta_{i_t} \geq p_t' \land i_t \in Q \land b_t = 0\right) = 0$, which means that \[\E\left[Z_5\right] \leq \E\left[\sum_{t> t_{\lambda}}p_t' - p_t\right].\]

Let $j'_t = \argmin_{j \in S_t} \LB_{jt}$, which means that
$p_t = \min_{i \in S_t}\left\{\min\left\{\theta_{i}, \LB_{it}\right\}\right\} = \min\left\{p_t', \LB_{j_t't}\right\}$. We also know that $p_t' = \theta_{j_t} \leq \theta_{j_t'}.$ Therefore, \[\E\left[Z_5\right] \leq \E\left[\sum_{t> t_{\lambda}}\max\left\{0, \theta_{j_t'} - \LB_{j_t't}\right\}\right].\]

For the remainder of our analysis, we will require the following events:
\begin{itemize}
    \item Let $\cE_1$ be the event that for all $t > t_{\lambda}$, $\left|\Phi_{i,t}\right| \geq \frac{1}{2}q_{\min}(t-1)$ for all $i \in S_t$. In Lemma~\ref{lemma:reviews_large_q}, we prove that if $q_{\min} > 2\lambda$, then $\Pr\left[\cE_1^c\right] \leq \frac{1}{T}.$
    \item Similarly, let $\cE_2$ be the event that for all $t > t_{\lambda}$, $\left|\Phi_{i,t}\right| \geq \frac{1}{4}\lambda(t-1)$ for all $i \in S_t$ such that $q_i \geq \frac{\lambda}{2}$. In Lemma~\ref{lemma:reviews_small_q}, we prove that if $q_{\min} \leq 2\lambda$, then $\Pr[\cE_2^c] \leq \frac{1}{T}.$
    \item Let $\cF$ be the event that for all $i \in [d]$ such that $q_i \leq \frac{\lambda}{2}$, we have that $\overline{q}_i < \frac{3\lambda}{4}$, which means that $i \not\in Q$. In Lemma~\ref{lem:F}, we prove that $\Pr[\cF^c] \leq \frac{1}{T}.$
    \item Let $\cH$ be the event that for all $t > t_{\lambda}$ and all $i \in S_t$, \[|\Phi_{it}|\theta_{i} \leq \sum_{v \in \Phi_{it}} v + \sqrt{\frac{1}{2}|\Phi_{it}|\ln(dT^2)}.\] In Lemma~\ref{lem:H}, we prove that $\Pr[\cH^c] \leq \frac{1}{T}.$
\end{itemize}

We now split our analysis into two cases depending on whether or not $q_{\min} > 2\lambda$. Suppose that $q_{\min} > 2\lambda.$ In this case, \begin{align*}\E\left[Z_5\right] &\leq \E\left[\left. \sum_{t> t_{\lambda}}\max\left\{0, \theta_{j_t'} - \LB_{j_t't}\right\} \, \right| \, \cE_1 \land \cH \right] + T\Pr[(\cE_1\land \cH)^c]\\
                &\leq \E\left[\left. \sum_{t> t_{\lambda}}\max\left\{0, \theta_{j_t'} - \LB_{j_t't}\right\} \, \right| \, \cE_1 \land \cH \right] + 2\\
                &= \E\left[\left. \sum_{t> t_{\lambda}}\max\left\{0, \theta_{j_t'} - \frac{1}{|\Phi_{j_t't}|}\sum_{v \in \Phi_{j_t't}}v + \sqrt{\frac{1}{2|\Phi_{j_t't}|}\ln \frac{1}{\eta}}\right\} \, \right| \, \cE_1 \land \cH \right] + 2.\end{align*}
Under events $\cE_1$ and $\cH$, \[\E[Z_5] \leq \E\left[\left. \sum_{t> t_{\lambda}} \sqrt{\frac{\ln(dT^2)}{2|\Phi_{j_t't}|}} + \sqrt{\frac{1}{2|\Phi_{j_t't}|}\ln \frac{1}{\eta}} \, \right| \, \cE_1 \land \cH \right] + 2\] and by definition of the event $\cE_1$,
\begin{align*}
                \E[Z_5]&\leq \E\left[\left. \sum_{t> t_{\lambda}} \sqrt{\frac{\ln(dT^2)}{2|\Phi_{j_t't}|}} + \sqrt{\frac{1}{2|\Phi_{j_t't}|}\ln \frac{1}{\eta}} \, \right| \, \cE_1 \land \cH \right] + 2\\
                &\leq \sum_{t = 2}^T \left(\sqrt{\frac{\ln(dT^2)}{q_{\min}(t-1)}} + \sqrt{\frac{1}{q_{\min}(t-1)}\ln \frac{1}{\eta}}\right) + 2\\
                &\leq 4\sqrt{\frac{T}{q_{\min}} \ln \frac{dT^2}{\eta}} + 2.
                \end{align*}

Meanwhile, suppose that $q_{\min} < 2\lambda.$ When $\cF$ happens, for all $t > t_{\lambda}$, $S_t \subseteq Q \subseteq \left\{i : q_i > \frac{\lambda}{2}\right\}$, so when both $\cE_2$ and $\cF$ happen, $|\Phi_{i,t}| \geq \frac{1}{4}\lambda(t-1)$ for all $t > t_{\lambda}$ and $i \in S_t.$ Therefore,
\begin{align*}\E\left[Z_5\right] &\leq \E\left[\left. \sum_{t> t_{\lambda}}\max\left\{0, \theta_{j_t'} - \LB_{j_t't}\right\} \, \right| \, \cE_2 \land \cF \land \cH \right] + T\Pr[(\cE_2 \land \cF \land \cH)^c]\\
                &\leq \E\left[\left. \sum_{t> t_{\lambda}}\max\left\{0, \theta_{j_t'} - \LB_{j_t't}\right\} \, \right| \, \cE_2 \land \cF \land \cH \right] + 3\\
                &= \E\left[\left. \sum_{t> t_{\lambda}}\max\left\{0, \theta_{j_t'} - \frac{1}{|\Phi_{j_t't}|}\sum_{v \in \Phi_{j_t't}}v + \sqrt{\frac{1}{2|\Phi_{j_t't}|}\ln \frac{1}{\eta}}\right\} \, \right| \, \cE_2 \land \cF \land \cH\right] + 3.\end{align*}
When $\cE_2$, $\cF$, and $\cH$ all happen, \[\E\left[Z_5\right] \leq \E\left[\left. \sum_{t> t_{\lambda}} \sqrt{\frac{\ln(dT^2)}{2|\Phi_{j_t't}|}} + \sqrt{\frac{1}{2|\Phi_{j_t't}|}\ln \frac{1}{\eta}} \, \right| \, \cE_2 \land \cF \land \cH \right] + 3\] and by definition of $\cE_2 \land \cF$, 
               \begin{align*}\E\left[Z_5\right] &\leq \E\left[\left. \sum_{t> t_{\lambda}} \sqrt{\frac{\ln(dT^2)}{2|\Phi_{j_t't}|}} + \sqrt{\frac{1}{2|\Phi_{j_t't}|}\ln \frac{1}{\eta}} \, \right| \, \cE_2 \land \cF \land \cH \right] + 3\\
                &\leq \sum_{t = 2}^T \left(\sqrt{\frac{2\ln(dT^2)}{\lambda(t-1)}} + \sqrt{\frac{2}{\lambda(t-1)}\ln \frac{1}{\eta}}\right) + 3\\
                &\leq 4\sqrt{\frac{2T}{\lambda} \ln \frac{dT^2}{\eta}} + 3.
                \end{align*}
\end{proof}

\A*
\begin{proof}
Recall that $\widehat{\mu}_{i,t}= \overline{\mu}_{i,t} + \rho_t$ and $\widecheck{\mu}_{i,t} = \overline{\mu}_{i,t} - \rho_t$ with \[\rho_t = \sqrt{\frac{\ln (dT^2)}{2\left(t-t_{\lambda}\right)}}.\] We also define the related quantities for all $i \in [d]$ and all $Q' \subseteq [d]$: \[\overline{\gamma}_{i,t}(Q')
            = \frac{1}{t-t_{\lambda}}\sum_{s = t_{\lambda} + 1}^{t} \theta_i \cdot \Ind\left(\theta_{i_s} \geq \theta_i \land i_s \in Q'\right),\] $\widehat{\gamma}_{i,t}(Q')= \overline{\gamma}_{i,t}(Q') + \rho_t$, and $\widecheck{\gamma}_{i,t}(Q') = \overline{\gamma}_{i,t}(Q') -\rho_t$. By a Hoeffding bound, for all $Q' \subseteq [d]$ and all $i \in [d]$, \[\Pr\left[\rev\left(\theta_{i}, Q'\right)  \not\in \left[\widecheck{\gamma}_{i,t}(Q'), \widehat{\gamma}_{i,t}(Q')\right]\right] \leq \frac{1}{dT^2}.\]

We claim that for any $i \in S_{t}$ and any $s > t_{\lambda}$, \begin{equation}\Ind(b_s = 1 \land \theta_{i_s} \geq \theta_{i} \land i_s \in Q) = \Ind\left(\theta_{i_s} \geq \theta_i \land i_s \in Q\right),\label{app:eqestimators_equal_small_q}\end{equation} which means that $\overline{\mu}_{i, t} = \overline{\gamma}_{i, t}(Q),$ $\widehat{\mu}_{i, t} = \widehat{\gamma}_{i, t}(Q)$, and $\widecheck{\mu}_{i, t} = \widecheck{\gamma}_{i, t}(Q)$. To see why, if $b_s = 1$, then clearly Equation~\eqref{app:eqestimators_equal_small_q} holds. Otherwise, suppose $b_s = 0$, in which case $\Ind(b_s = 1 \land \theta_{i_s} \geq \theta_{i} \land i_s \in Q) = 0$. Then $i_s \not \in S_s$ because any buyer in $S_s$ will always buy by definition of the pricing rule. Let $j_s = \min\left\{j \in S_s\right\}.$ Since $S_s$ contains every element in $Q$ larger than $j_s$, we know that either:
            \begin{enumerate}
            \item $i_s \not\in Q$, in which case $\Ind\left(\theta_{i_s} \geq \theta_i \land i_s \in Q\right) = 0$, or 
            \item $i_s \in Q$ but $i_s \not\in S_s$, which means that $\theta_{i_s} < \theta_{j_s}$. Since $i \in S_t$, it must be that $i \in S_s$, so $\theta_{i_s} < \theta_{j_s} \leq \theta_{i}$. In this case, $\Ind\left(\theta_{i_s} \geq \theta_i \land i_s \in Q\right) = 0$ as well.
            \end{enumerate}
Therefore, Equation~\eqref{app:eqestimators_equal_small_q} holds.

The fact that $\overline{\mu}_{i, t} = \overline{\gamma}_{i, t}(Q),$ $\widehat{\mu}_{i, t} = \widehat{\gamma}_{i, t}(Q)$, and $\widecheck{\mu}_{i, t} = \widecheck{\gamma}_{i, t}(Q)$ for all $i \in S_t$ implies that \begin{align}
    \Pr[\cA_t^c] &= \Pr\left(\exists i \in S_{t} \text{ s.t. }\rev\left(\theta_{i}, Q\right) \not\in \left[\widecheck{\mu}_{i,t} , \widehat{\mu}_{i,t} \right]\right)\nonumber\\
    &\leq \Pr\left(\exists i \in [d] \text{ s.t. }\rev\left(\theta_{i}, Q\right) \not\in \left[\widecheck{\gamma}_{i,t}(Q) , \widehat{\gamma}_{i,t}(Q) \right]\right)\nonumber\\
    &\leq \sum_{i = 1}^d\Pr\left(\rev\left(\theta_{i}, Q\right) \not\in \left[\widecheck{\gamma}_{i,t}(Q) , \widehat{\gamma}_{i,t}(Q) \right]\right).\label{app:eqsum_d}
\end{align}

The set $Q$ is a random variable, so we must condition on it to bound Equation~\eqref{app:eqsum_d}: \begin{align*}&\Pr\left(\rev\left(\theta_{i}, Q\right) \not\in \left[\widecheck{\gamma}_{i,t}(Q) , \widehat{\gamma}_{i,t}(Q) \right]\right)\\
= &\sum_{Q' \subseteq [d]}\Pr\left(\rev\left(\theta_{i}, Q'\right) \not\in \left[\widecheck{\gamma}_{i,t}(Q') , \widehat{\gamma}_{i,t}(Q') \right] \mid Q = Q'\right)\Pr[Q = Q'].\end{align*}
Since the event that $Q = Q'$ and the event that $\rev\left(\theta_{i}, Q'\right) \not\in \left[\widecheck{\gamma}_{i,t}(Q') , \widehat{\gamma}_{i,t}(Q') \right]$ depend on disjoint timesteps, the two events are independent. Therefore, \begin{align*}\Pr\left(\rev\left(\theta_{i}, Q\right) \not\in \left[\widecheck{\gamma}_{i,t}(Q) , \widehat{\gamma}_{i,t}(Q) \right]\right) &=\sum_{Q' \subseteq [d]}\Pr\left(\rev\left(\theta_{i}, Q'\right) \not\in \left[\widecheck{\gamma}_{i,t}(Q') , \widehat{\gamma}_{i,t}(Q') \right]\right)\Pr\left[Q = Q'\right]\\
&\leq \frac{1}{dT^2}\sum_{Q' \subseteq [d]}\Pr\left[Q = Q'\right]\\
&= \frac{1}{dT^2},\end{align*} so the result follows from Equation~\eqref{app:eqsum_d}.
\end{proof}

The next lemma shows that for more common types with $q_i \geq \lambda$, the fraction of times $\overline{q}_i$ that that type appears during Algorithm~\ref{alg:type_elim} is large enough that $i$ is added to $Q$.

\begin{lemma}\label{lem:G}
    Let $\cG$ be the event that for all $i$ such that $q_i \geq \lambda$, we have that $\overline{q}_i \geq \frac{3\lambda}{4}$. Then $\Pr[\cG^c] \leq \frac{1}{T}.$
\end{lemma}

\begin{proof}
Fix an index $i$ such that $q_i \geq \lambda.$ Then \[\Pr\left[\overline{q}_i < \frac{3\lambda}{4}\right] = \Pr\left[\sum_{t=1}^{t_{\lambda}} \Ind(i_t = i) < \lambda t_{\lambda} \cdot \frac{3}{4}\right] \leq \exp\left(-\frac{\lambda t_{\lambda}}{32}\right) \leq \frac{1}{dT}.\]
The lemma the follows by a union bound over all $i \in [d].$
\end{proof}

The next lemma proves that the expected revenue (with respect to agents in $Q$) of the smallest active price $\min\left\{\theta_i : i \in S_t\right\}$ is converging to the optimal revenue $\rev(p^*(Q), Q)$ as $t$ grows. Later in the analysis, we will show---at a high level---that since the algorithm sets a price within a neighborhood of $\min\left\{\theta_i : i \in S_t\right\}$, its revenue is converging to that of $p^*(Q).$
For this next lemma, recall that $p^*(Q) = \theta_{i_Q}$ for some $i_Q \in Q$. The proof is similar to that of standard successive arm elimination algorithms~\citep[e.g.,][]{Zhao20:Stochastic}.

\begin{lemma}\label{lemma:SE}
    For all $t > t_{\lambda}$, let $j_t = \min\{j \in S_{t}\}$. Let $\cB_t$ be the event that:
    \begin{enumerate}
    \item $i_Q \in S_{t}$ and
    \item $\rev(p^*(Q), Q) - \rev\left(\theta_{j_t}, Q\right) \leq 4\rho_{t-1}$ (where $\rho_{t_{\lambda}} = 1$).
    \end{enumerate}
    Also, let $\cC_t = \bigcap_{s = t_{\lambda} + 1}^t \cB_s.$
    Then $\Pr\left[\cC_t^c\right] \leq \frac{1}{T}.$
\end{lemma}

\begin{proof}
We begin by partitioning $\cC_t^c$ into the disjoint events \[\cC_t^c = \cC_{t_{\lambda}+1}^c \cup \left(\cC_{t_{\lambda}+1} \cap \cB_{t_{\lambda}+2}^c\right) \cup \cdots \cup 
        \left(\cC_{t-1} \cap \cB_t^c\right).\] Since these events are disjoint, \begin{equation}\Pr\left[\cC_t^c\right] = \Pr\left[\cC_{t_{\lambda}+1}^c\right] + \Pr\left[\cC_{t_{\lambda}+1} \cap \cB_{t_{\lambda}+2}^c\right] + \cdots + 
       \Pr\left[\cC_{t-1} \cap \cB_t^c\right].\label{appeq:disjoint}\end{equation} Beginning with the first summand, $\Pr\left[\cC_{t_{\lambda}+1}^c\right] = \Pr\left[\cB_{t_{\lambda}+1}^c\right] = 0$ because $S_{t_{\lambda}} = Q$, so $i_Q \in S_{t_{\lambda}}$, and $4\rho_{t_{\lambda}} > 1.$

       Next, for $s > t_{\lambda} + 1$, \begin{equation}\Pr\left[\cC_{s-1} \cap \cB^c_s\right] = \Pr\left[\bigcap_{s' = t_{\lambda}+1}^{s-1}\cB_{s'} \cap \cB^c_s\right] \leq \Pr\left[\cB_{s-1} \cap \cB^c_s\right].\label{appeq:removeC}\end{equation} We will prove that $\cB_{s-1} \cap \cB_s^c$ implies $\cA_{s-1}^c$, which will allow us to apply Lemma~\ref{lem:A}.

\begin{claim}\label{claim:BimpliesA}
The event $\cB_{s-1} \cap \cB_s^c$ implies $\cA_{s-1}^c$.
\end{claim}
\begin{proof}Proof of Claim~\ref{claim:BimpliesA}]
       First suppose $\cB_{s-1}$ happens and $i_Q \not\in S_{s}$ (so $\cB_s^c$ happens).  Since $\cB_{s-1}$ happens, we know that $i_Q \in S_{ s-1}$ but since $i_Q \not\in S_{s}$, it must be that $i_Q$ was eliminated at the end of round $s-1$. This means that $\widehat{\mu}_{i_Q, s-1} < \max_{k \in S_{s-1}} \widecheck{\mu}_{k, s-1}$. Let $k' = \argmax_{k \in S_{s-1}} \widecheck{\mu}_{k, s-1}$. Then \begin{align}\rev(p^*(Q), Q) - \widecheck{\mu}_{i_Q, s-1} &\geq \rev\left(\theta_{k'}, Q\right) - \widecheck{\mu}_{i_Q, s-1}\nonumber\\
                &= \rev\left(\theta_{k'}, Q\right) - \widehat{\mu}_{i_Q, s-1} + 2\rho_{s-1}\nonumber\\
                &> \rev\left(\theta_{k'}, Q\right) - \widecheck{\mu}_{k',s-1} + 2\rho_{s-1}.\label{appeq:k'}
                \end{align} Suppose that $\rev\left(\theta_{k'}, Q\right) \geq \widecheck{\mu}_{k', s-1}$. Then Equation~\eqref{appeq:k'} implies that \[2\rho_{s-1} < \rev(p^*(Q), Q) - \widecheck{\mu}_{i_Q, s-1} = \rev(p^*(Q), Q) - (\widehat{\mu}_{i_Q, s-1} - 2\rho_{s-1})\] so $\rev(p^*(Q), Q) > \widehat{\mu}_{i_Q,  s-1}$. Therefore, either $\rev\left(\theta_{k'}, Q\right) < \widecheck{\mu}_{k', s-1}$ or $\rev(p^*(Q), Q) > \widehat{\mu}_{i_Q, s-1}$, which means that $\cA_{s-1}^c$ happens.

Meanwhile, suppose $\cB_{s-1}$ happens and $i_Q \in S_{s}$ but $\rev(p^*(Q), Q) - \rev\left(\theta_{j_s}, Q\right) > 4\rho_{s-1}$ (so $\cB_s^c$ happens). Then \begin{equation}\rev(p^*(Q), Q) - \widecheck{\mu}_{i_Q, s-1} + \widehat{\mu}_{j_s, s-1} - \rev\left(\theta_{j_s}, Q\right) > \widehat{\mu}_{j_s,  s-1} - \widecheck{\mu}_{i_Q, s-1} + 4\rho_{s-1}.\label{appeq:meanwhile}\end{equation} Again, let $k' = \argmax_{k \in S_{s-1}} \widecheck{\mu}_{k, s-1}$.
Since $j_s \in S_s$, it must be that $\widehat{\mu}_{j_s, s-1} \geq \widecheck{\mu}_{k', s-1}$, or else $j_s$ would have been eliminated at the end of round $s-1$.
Combining this fact with Equation~\eqref{appeq:meanwhile}, we have that \[\rev(p^*(Q), Q) - \widecheck{\mu}_{i_Q, s-1} + \widehat{\mu}_{j_s,  s-1} - \rev\left(\theta_{j_s}, Q\right) > \widecheck{\mu}_{k', s-1} - \widecheck{\mu}_{k', s-1} + 4\rho_{s-1} = 4\rho_{s-1}.\]
This means that either:
                \begin{enumerate}
                \item $2\rho_{s-1} < \rev(p^*(Q), Q) - \widecheck{\mu}_{i_Q, s-1} = \rev(p^*(Q), Q) - \widehat{\mu}_{i_Q, s-1} + 2\rho_{s-1}$, or in other words $\widehat{\mu}_{i_Q, s-1} < \rev(p^*(Q), Q)$, meaning $\cA_{s-1}^c$ happens, or
                \item $2\rho_{s-1} < \widehat{\mu}_{j_s,  s-1} - \rev\left(\theta_{j_s}, Q\right) = \widecheck{\mu}_{j_s,  s-1} + 2\rho_{s-1} - \rev\left(\theta_{j_s}, Q\right)$, or in other words, $\rev\left(\theta_{j_s}, Q\right) < \widecheck{\mu}_{j_s,  s-1}$, meaning $\cA_{s-1}^c$ happens.
                \end{enumerate}
Therefore, the claim holds.
        \end{proof}
Claim~\ref{claim:BimpliesA}, Equation~\eqref{appeq:removeC}, and Lemma~\ref{lem:A} imply that $\Pr\left[\cC_{s-1} \cap \cB^c_s\right] \leq \Pr\left[A_{s-1}^c\right] \leq \frac{1}{T^2}$, so by Equation~\eqref{appeq:disjoint}, we have that $\Pr\left[\cC_t^c\right] < \frac{1}{T}$.
\end{proof}

The next lemma will prove that for all rounds $t>t_\lambda$ of Algorithm~\ref{alg:main} and all active types $i \in S_t$, there are a non-trivial number of reviews by buyers of type $i$. The following lemma holds when $q_{\min} > 2\lambda$, and Lemma~\ref{lemma:reviews_small_q} holds when $q_{\min} \leq 2\lambda$.

\begin{lemma}\label{lemma:reviews_large_q}
Suppose that $q_{\min} > 2\lambda$.  Let $\cE_1$ be the event that on each round $t > t_{\lambda}$, $\left|\Phi_{i,t}\right| \geq \frac{1}{2}q_{\min}(t-1)$ and all $i \in S_t$. Then $\Pr[\cE_1^c] \leq \frac{1}{T}.$
\end{lemma}

\begin{proof}
Fix any $t > t_{\lambda}$. We will show that \[\Pr\left[\exists i \in S_t \text{ such that  }|\Phi_{i,t}| < \frac{1}{2}q_{\min}(t-1)\right] \leq \frac{1}{T^2}.\] By definition, $\left|\Phi_{i,t}\right| = \sum_{s = 1}^{t-1} \Ind(b_s = 1 \land i_s = i).$ If $\left|\Phi_{i,t}\right|$ were equal to $\sum_{s = 1}^{t-1} \Ind(i_s = i),$ then the claim would hold immediately by a Chernoff bound. However, we do not know at each round $s$ whether $i_s = i$ provided the buyer does not make a purchase.
Therefore, we also define the random variable $X_{i,t} = \sum_{s = 1}^{t-1} \Ind(i_s = i)$. We claim that for all $i \in S_t$, $\left|\Phi_{i,t}\right| = X_{i,t}$. This is because we know that $i \in S_s$ for all $s \leq t$ and by definition of the pricing rule, if $i_s = i$, then $b_s = 1$.

Therefore, \begin{align}&\Pr\left[\exists t > t_{\lambda}, \exists i \in S_t \text{ such that  }\left|\Phi_{i,t}\right| < \frac{1}{2}q_{\min}(t-1)\right]\nonumber\\
        = \, &\Pr\left[\exists t > t_{\lambda}, \exists i \in S_t \text{ such that }X_{i,t} < \frac{1}{2}q_{\min}(t-1)\right]\nonumber\\
        \leq\, & \Pr\left[\exists t > t_{\lambda}, \exists i \in [d] \text{ such that  }X_{i,t} < \frac{1}{2}q_{\min}(t-1)\right]\nonumber\\
        \leq\, & \sum_{i = 1}^d\sum_{t =t_{\lambda}+1}^T\Pr\left[X_{i,t} < \frac{1}{2}q_{\min}(t-1)\right].\label{appeq:X}\end{align} By a Chernoff bound, \begin{align*}\Pr\left[X_{i,t} \leq \frac{1}{2}q_{\min}(t-1)\right] &\leq \Pr\left[X_{i,t} \leq \frac{1}{2}q_i(t-1)\right]\\
        &\leq \exp\left(-\frac{q_i(t-1)}{8}\right)\\
        &\leq \exp\left(-\frac{q_{\min}(t-1)}{8}\right)\\
        &\leq \exp\left(-\frac{\lambda(t-1)}{4}\right)\\
        &\leq \frac{1}{dT^2}.\end{align*} The lemma now follows from Equation~\eqref{appeq:X}.
\end{proof}

We now prove a similar result for the case where $q_{\min} \leq 2\lambda$.

\begin{lemma}\label{lemma:reviews_small_q}
Suppose that $q_{\min} \leq 2\lambda$.  Let $\cE_2$ be the event that for all $t > t_{\lambda}$, $\left|\Phi_{i,t}\right| \geq \frac{1}{4}\lambda(t-1)$ for all $i \in S_t$ such that $q_i \geq \frac{\lambda}{2}$. Then $\Pr[\cE_2^c] \leq \frac{1}{T}.$
\end{lemma}

\begin{proof}
Let $Q_0 = \left\{i : q_i \geq \frac{\lambda}{2}\right\}.$ Fix any $t > t_{\lambda}$. We will show that \[\Pr\left[\exists i \in S_t \cap Q_0 \text{ such that  }\left|\Phi_{i,t}\right| < \frac{1}{4}\lambda(t-1)\right] \leq \frac{1}{T^2}.\] By definition, $\left|\Phi_{i,t}\right| = \sum_{s = 1}^{t-1} \Ind(b_s = 1 \land i_s = i).$ As in the proof of Lemma~\ref{lemma:reviews_large_q}, we define the random variable $X_{i,t} = \sum_{s = 1}^{t-1} \Ind(\vec{x}_s = \vec{e}_i)$. As in that proof, for all $i \in S_t$, $\left|\Phi_{i,t}\right| = X_{i,t}$ (this is because we know that $i \in S_s$ for all $s \leq t$ and by definition of the pricing rule, if $i_s = i$, then $b_s = 1$.).

Therefore, \begin{align}&\Pr\left[\exists t > t_{\lambda}, \exists i \in S_t \cap Q_0 \text{ such that  }|\Phi_{i,t}| < \frac{1}{4}\lambda(t-1)\right]\nonumber\\
        =\, &\Pr\left[\exists t > t_{\lambda}, \exists i \in S_t \cap Q_0 \text{ such that  }X_{i,t} < \frac{1}{4}\lambda(t-1)\right]\nonumber\\
        \leq\, &\Pr\left[\exists t > t_{\lambda}, \exists i \in Q_0 \text{ such that  }X_{i,t} < \frac{1}{4}\lambda(t-1)\right]\nonumber\\
        \leq\, & \sum_{i \in Q_0}\sum_{t = t_{\lambda+1}}^T\Pr\left[X_{i,t} < \frac{1}{4}\lambda(t-1)\right]\label{appeq:Q0}.\end{align}
        By a Chernoff bound, for any $i \in Q_0$, \begin{align*}\Pr\left[X_{i,t} < \frac{1}{4}\lambda(t-1)\right] &\leq \Pr\left[X_{i,t} < \frac{1}{2}q_i(t-1)\right]\\
        &\leq \exp\left(-\frac{q_i(t-1)}{8}\right)\\
        &\leq \exp\left(-\frac{\lambda(t-1)}{16}\right)\\
        &\leq \frac{1}{dT^2}.\end{align*}
        The lemma therefore follows from Equation~\eqref{appeq:Q0}.
\end{proof}

We next observe that for all very rare types with $q_i \leq \lambda/2$, the fraction of times $\overline{q}_i$ that that type appears during Algorithm~\ref{alg:type_elim} is small. Therefore, $i$ is not added to the set $Q$ and is ignored for the remainder of the algorithm. 

\begin{lemma}\label{lem:F}
    Let $\cF$ be the event that for all $i \in [d]$ such that $q_i \leq \frac{\lambda}{2}$, we have that $\overline{q}_i \leq \frac{3\lambda}{4}$. Then $\Pr[\cF^c] \leq \frac{1}{T}.$
\end{lemma}

\begin{proof}
Fix an index $i$ such that $q_i \leq \frac{\lambda}{2}.$ Then \[\Pr\left[\overline{q}_i \geq \frac{3\lambda}{4}\right] = \Pr\left[\sum_{t=1}^{t_{\lambda}} \Ind(i_t = i) \geq \frac{\lambda t_{\lambda}}{2} \cdot \frac{3}{2}\right] \leq \exp\left(-\frac{\lambda t_{\lambda}}{24}\right) = \frac{1}{dT}.\]
The lemma the follows by a union bound over all $i \in [d].$
\end{proof}

Our final lemma proves that for all active types $i \in S_t$, the average reviews of agents with this type is close to the true \emph{ex-ante} value $\theta_i.$ This helps us ensure that the price we set is not too low.

\begin{lemma}\label{lem:H}
Let $\cH$ be the event that for all $t > t_{\lambda}$ and all $i \in S_t$, \[|\Phi_{it}|\theta_{i} \leq \sum_{v \in \Phi_{it}} v + \sqrt{\frac{1}{2}|\Phi_{it}|\ln(dT^2)}.\] Then $\Pr[\cH^c] \leq \frac{1}{T}.$
\end{lemma}

\begin{proof}
Fix any $t > t_{\lambda}$. Let $v_1, \dots, v_{t-1}$ be the buyers' \emph{ex-post} values (which are defined even if the buyer didn't buy on a particular round $s$ as $v_s \sim \cD_{i_s}$). For each $i \in [d]$, let $R_{it} = \left\{s < t : i_s = i\right\}$ be the set of rounds in which the buyer had type $i$. Since any buyer $i \in S_s$ will buy if $i_s = i$, we have that $|\Phi_{it}| = |R_{it}|$ and  \[\sum_{v \in \Phi_{it}} v = \sum_{s \in R_{it}}v_s.\] Therefore, \begin{align}
    &\Pr\left[\exists i \in S_t \text{ such that } |\Phi_{it}|\theta_{i} > \sum_{v \in \Phi_{it}} v + \sqrt{\frac{1}{2}|\Phi_{it}|\ln(dT^2)}\right]\nonumber\\
    \leq \, &\Pr\left[\exists i \in [d] \text{ such that } |R_{it}|\theta_{i} > \sum_{s \in R_{it}} v_s + \sqrt{\frac{1}{2}|R_{it}|\ln(dT^2)}\right]\nonumber\\
    \leq \, &\sum_{i =1}^d \Pr\left[|R_{it}|\theta_{i} > \sum_{s \in R_{it}} v_s + \sqrt{\frac{1}{2}|R_{it}|\ln(dT^2)}\right]\nonumber\\
    =\, &\sum_{i = 1}^d \sum_{R \subseteq [t-1]}\Pr\left[\left. |R|\theta_{i} > \sum_{s \in R} v_s + \sqrt{\frac{1}{2}|R|\ln(dT^2)} \, \right| \, R_{it} = R\right]\Pr[ R_{it} = R].\label{appeq:R}\end{align}
For any $s \in R$, $\E\left[v_s \mid R_{it} = R\right] = \theta_i$. Therefore, \[\Pr\left[\left. |R|\theta_{i} > \sum_{s \in R} v_s + \sqrt{\frac{1}{2}|R|\ln(dT^2)} \, \right| \, R_{it} = R\right] \leq \frac{1}{dT^2}.\] The lemma therefore follows from Equation~\eqref{appeq:R} and a union bound over all rounds $t > t_{\lambda}$.
\end{proof}

\end{document}